\documentclass[onecolumn, 11pt,letterpaper]{article}

\usepackage{graphicx}     			
\usepackage{multirow}               		
\usepackage{amssymb}                	
\usepackage{latexsym}               		
\usepackage{alltt}                  		
\usepackage{amsgen}                 	
\usepackage{amstext}               		
\usepackage{amscd}                  		
\usepackage{algorithmic}
\usepackage{url}
\usepackage{enumerate}
\usepackage{amsthm}
\usepackage{amsmath}
\usepackage[left=1in,top=1in,right=1in,bottom = 1in, nohead]{geometry}

\usepackage{bm}                     	
\usepackage{upgreek}                

\usepackage{xspace}            
\usepackage[caption=false]{subfig}            

\usepackage{color}                 
\usepackage{colortbl}             

\usepackage{tabularx}		

\usepackage{aliascnt}  		

\usepackage{float}                 
\usepackage{ifpdf}
    \ifpdf                                                                  
        \usepackage[breaklinks=true]{hyperref}
        \hypersetup{
            colorlinks=false,               
            citebordercolor={0 1 0},        
            urlbordercolor={0 0 1},         
            linkbordercolor={1 0.2 0.2},    
            pdfborder={0 0 1.2},            
        }
    \else
        \usepackage[hypertex, breaklinks=true]{hyperref}       
        \hypersetup{
            colorlinks=true,               
            anchorcolor=yellow,
            linkcolor=green,
            filecolor=green,
            urlcolor=blue,
            citecolor=red,
        }
    \fi
    \hypersetup{                                                            
        pdfpagemode=UseNone,               
        pdfstartview=Fit,
        pdfpagelayout=SinglePage,       
        bookmarks=true,
        bookmarksnumbered=true,
        bookmarksopenlevel=1,            
        bookmarkstype=toc,
        pdftex=true,
        unicode,
        plainpages=false,           
        hypertexnames=true,         
        pdfpagelabels=true,         
        hyperindex=true,
        naturalnames=false,                
    }

\newtheorem{theorem}{Theorem}         	

\newaliascnt{lemma}{theorem}				
\newtheorem{lemma}[lemma]{Lemma}              	
\aliascntresetthe{lemma}  					
\newaliascnt{conjecture}{theorem}			
\aliascntresetthe{conjecture}  				
\newaliascnt{remark}{theorem}				

\aliascntresetthe{remark}  					
\newaliascnt{corollary}{theorem}			
\newtheorem{corollary}[corollary]{Corollary}      
\aliascntresetthe{corollary}  				
\newaliascnt{definition}{theorem}			
\newtheorem{definition}[definition]{Definition}    
\aliascntresetthe{definition}  				
\newaliascnt{proposition}{theorem}			
\newtheorem{proposition}[proposition]{Proposition}  
\aliascntresetthe{proposition}  				
\newaliascnt{example}{theorem}			
\newtheorem{example}[example]{Example}  	
\aliascntresetthe{example}  				

  {\refstepcounter{theorem}\trivlist\item       
  [\hskip\labelsep\bf\sf Policy Query \thetheorem]}
  {\endtrivlist}                                

\let\orgautoref\autoref                         		

\renewcommand{\autoref}[1]{
    \def\equationautorefname{Eq.}
    \def\figureautorefname{Fig.}%
    \def\subfigureautorefname{Fig.}%
    \def\lemmaautorefname{Lemma}%
    \def\conjectureautorefname{Conjecture}%
    \def\remarkautorefname{Remark}%
    \def\propositionautorefname{Prop.}%
    \def\corollaryautorefname{Corollary}%
    \def\definitionautorefname{Def.}%
    \def\sectionautorefname{Sect.}%
    \def\subsectionautorefname{Sect.}%
    \def\subsubsectionautorefname{Section}%
    \def\exampleautorefname{Example}%
    \orgautoref{#1}%
}

\newcommand{\specificref}[2]{\hyperref[#2]{#1~\ref*{#2}}}			

\newcommand{\EE}[1]{{\textbf{\textup{\textrm{E}}}} \left( #1 \right)}        
\newcommand{\VV}[1]{{\textup{\textbf{Var}}} \left( #1 \right)}      
    \newcommand{\PP}[1]{{\textup{\textrm{Pr}}} \left( #1 \right)}                     

\newcommand{\R}{\mathbb{R}}             
\renewcommand{\epsilon}{\varepsilon}    

\definecolor{gray}{rgb}{0.5,0.5,0.5}
\definecolor{niceblue}{rgb}{.8,.85,1}

\newcommand\independent{\protect\mathpalette{\protect\independenT}{\perp}}			
\def\independenT#1#2{\mathrel{\rlap{$#1#2$}\mkern2mu{#1#2}}}

\newcommand{\set}[1]{\{#1\}}                    
\newcommand{\makeop}[2]                         
  {\ifx#2.\def\next##1{}\else\escapechar=-1     
  \def\next##1{\escapechar=92\def#2{#1}}        
  \expandafter\next\expandafter{\string#2}      
  \let\next\makeop\fi\next{#1}}                 
\makeop{{\cal #1}}                              
  \W  .             %

\def \var(#1){{\bf #1}}

\def\AddSpace#1{\ifcat#1a\ \fi#1} 

\newcommand{\silentreminder}[1]{}

\def \up(#1){[#1)}
\def \down(#1){(#1]}

\def \series(#1,#2){#1_1, \dots \; #1_{#2}}
\def \serieszero(#1,#2){#1_0, #1_1, \dots \; #1_{#2}}

\def \para(#1){{\vspace{1ex}\noindent\small\bf #1\hspace{1ex}}}
\def \myem(#1){{\vspace{1ex}\noindent\small\em #1\hspace{1ex}}}

\newcommand{\eat}[1]{}

\renewenvironment{thebibliography}[1]{%
    \begin{oldthebibliography}{#1}%
    	\small
        \setlength{\parskip}{0.1ex}
        \setlength{\itemsep}{0.0ex}%
}%
{%
\end{oldthebibliography}%
}

\newcommand{\bbb}[1]{{\mathbf{#1}}}

\usepackage{amsthm}

\date{}
\begin{document}
\title{A Theory of Pricing Private Data}



\author{Chao Li$^1$ \and Daniel Yang Li$^2$ \and Gerome Miklau$^{1,3}$ \and Dan Suciu$^2$ \\
\and $^1$University of Massachusetts \\Amherst, MA, USA\\ \{chaoli, miklau\}@cs.umass.edu
\and $^2$University of Washington \\Seattle, WA, USA\\ \{dyli, suciu\}@cs.washington.edu
\and $^3${INRIA} \\ Saclay, France
}

\date{}

\maketitle 


\begin{abstract}

Personal data has value to both its owner and to institutions who would like to analyze it.  Privacy mechanisms protect the owner's data while releasing to analysts noisy versions of aggregate query results.  But such strict protections of individual's data have not yet found wide use in practice.  Instead, Internet companies, for example, commonly provide free services in return for valuable sensitive information from users, which they exploit and sometimes sell to third parties.

As the awareness of the value of the personal data increases, so has the drive to compensate the end user for her private information.  The idea of monetizing private data can improve over the narrower view of hiding private data, since it empowers individuals to control their data through financial means.

In this paper we propose a theoretical framework for assigning prices to noisy query answers, as a function of their accuracy, and for dividing the price amongst data owners who deserve compensation for their loss of privacy.  Our framework adopts and extends key principles from both differential privacy and query pricing in data markets.  We identify essential properties of the price function and micro-payments, and characterize valid solutions.

\end{abstract}


\newpage


\section{Introduction}

\label{sec:intro}

Personal data has value to both its owner and to institutions who would like to analyze it.  The interests of individuals and institutions with respect to personal data are often at odds and a rich literature on privacy-preserving data publishing techniques~\cite{DBLP:journals/csur/FungWCY10} has tried to devise technical methods for negotiating these competing interests.  Broadly construed, privacy refers to an individual's right to control how her private data will be used, and was originally phrased as an individual's right to be protected against gossip and slander~\cite{danezis:2010}.  Research on privacy-preserving data publishing has focused more narrowly on privacy as data confidentiality.  For example, in perturbation-based data privacy, the goal is to protect an individual's personal data while releasing to legitimate users the result of aggregate computations over a large population~\cite{DBLP:journals/cacm/Dwork11}.

To date, this goal has remained elusive.  One important result from that line of work is that any mechanism providing reasonable privacy must strictly limit the number of query answers that can be accurately released~\cite{DBLP:conf/pods/DinurN03}, thus imposing a strict {\em privacy budget} for any legitimate user of the data~\cite{DBLP:journals/cacm/McSherry10}.  Researchers are actively investigating formal notions of privacy and their implications for effective data analysis.  Yet, with rare exception \cite{Kifer08Privacy:}, perturbation-based privacy mechanisms have not been deployed in practice.

Instead, many Internet companies have followed a simple formula to acquire personal data.  They offer a free service, attract users who
provide their data, and then monetize the personal data by selling
it, or by selling information derived from it, to third parties.
A recent study by JPMorgan Chase~\cite{nyt}
found that each unique user is worth approximately \$4 to Facebook and \$24 to Google.

Currently, many users are willing to provide their private data in
return for access to online services.  But as individuals become more
aware of the use of their data by corporate entities, of the potential
consequences of disclosure, and of the ultimate value of their
personal data, there has been a drive to compensate them
directly~\cite{11Personal}.
In fact, startup companies are currently developing infrastructure to support this trend.  For example, \url{www.personal.com}
creates personal data vaults, each of which may contain thousands of data points about its users.  Businesses pay for this data, and the data owners are appropriately compensated.

Monetizing private data is an improvement over the narrow view of privacy as data confidentiality because it empowers individuals to control their data through financial means.  In this paper we propose a framework for assigning prices to queries in order to compensate the data owners for their loss of privacy.  Our framework borrows from, and extends, key principles from both differential privacy~\cite{DBLP:journals/cacm/Dwork11} and data markets~\cite{DBLP:conf/pods/KoutrisUBHS12,lipricing}.
There are three actors in our setting: individuals, or data {\em owners}, contribute their personal data; a {\em buyer} submits an aggregate query over many owners' data; and a {\em market maker}, trusted to answer queries on behalf of owners, charges the buyer and compensates the owners.  Our framework makes three important connections:

\medskip

{\bf Perturbation and Price} In response to a buyer's query, the
market maker computes the true query answer, adds random noise, and
returns a perturbed result.  While under differential privacy
perturbation is always necessary, here query answers could be sold
unperturbed, but the price would be high because each data owner
contributing to an aggregate query needs to be compensated.  By adding
perturbation to the query answer, the price can be lowered: the more
perturbation, the lower the price.  The buyer specifies how much
accuracy he is willing to pay for when issuing the query.  Unperturbed
query answers are very expensive, but at the other extreme, query
answers are almost free if the noise added is the same as in
differential privacy~\cite{DBLP:journals/cacm/Dwork11} with
conservative privacy parameters. The relationship between the accuracy
of a query result and its cost depends on the query and the
preferences of contributing data owners.  Formalizing this
relationship is one of the goals of this paper.

{\bf Arbitrage and Perturbation} Arbitrage is an undesirable property
of a set of priced queries that allows a buyer to obtain the answer to
a query more cheaply than its advertised price by deriving the answer
from a less expensive alternative set of queries.  As a simple
example, suppose that a given query is sold with two options for
perturbation, measured by variance: a variance of 10 for \$5 and a
variance of 1 for \$200.  A savvy buyer who seeks a variance of 1
would never pay \$200. Instead, he would purchase the first query 10
times, receive 10 noisy answers, and compute their average.  Since the
noise is added independently, the variance of the resulting average is
1, and the total cost is only \$50.  Arbitrage opportunities result
from inconsistencies in the pricing of queries which must be avoided
and perturbing query answers makes this significantly more
challenging.  Avoiding arbitrage in data markets has been considered
before only in the absence of perturbation
\cite{DBLP:journals/pvldb/BalazinskaHS11,DBLP:conf/pods/KoutrisUBHS12,lipricing}. Formalizing
arbitrage for noisy queries is a second goal of this paper.  While, in
theory, achieving arbitrage-freeness requires imposing a lower bound
on the ratio between the price of low accuracy and high accuracy
queries, we will show that it is possible to design quite flexible
arbitrage-free pricing functions.


{\bf Privacy-loss and Payments} Given a randomized mechanism for
answering a query $q$, a common measure of privacy loss to an
individual is defined by differential privacy: it is the maximum ratio
between the probability of returning some fixed output with and
without that individual's data.  Differential privacy imposes a bound
of $e^\varepsilon$ on this quantity, where $\varepsilon$ is a small
constant, presumed acceptable to all individuals in the population.
Our framework contrasts with this in several ways.  First, the privacy
loss is not limited a priori, but depends on the buyer's request.  If
the buyer asks for a query with low variance, then the privacy loss to
(at least some) individuals will be high.  These data owners must be
compensated for their privacy loss through the buyer's payment.  At an
extreme, if the query answer is exact (unperturbed), then the privacy
loss to some individuals is total, and they must be compensated
appropriately.  Also, we allow each data owner to value their privacy
loss separately, by demanding greater or lesser payments.  Formalizing
the relationship between privacy loss and payments to the data owners
is a third goal of this paper.


By charging buyers for access to private data we overcome a
fundamental limitation of perturbation-based privacy preserving
mechanisms, namely the privacy budget.  This term refers to a limit on
the quantity and/or accuracy of queries that any buyer can ask, in order to prevent an unacceptable disclosure of the data.  For example, if a differentially-private mechanism adds Laplacian noise with variance $v$, then by asking the same query $n$ times the buyer can reduce the variance to $v/n$.  Even if queries are restricted to aggregate queries, there exist sequences of queries that can reveal the private data for most individuals in the database~\cite{DBLP:conf/pods/DinurN03} and enforcing the privacy budget must prevent this.  In contrast, when private data is priced, full disclosure is possible only if the buyer pays a high price.  For example, in order to reduce the variance to $v/n$, the buyer would have to purchase the query $n$ times, thus paying $n$ times more than for a single query.  In order to perform the attacks in~\cite{DBLP:conf/pods/DinurN03} he would have to pay for (roughly) $n\log^2 n$ queries.

Thus, the burden of the market maker is no longer to guard the privacy
budget, but instead to ensure that prices are set such that, whatever
disclosure is obtained by the buyer, all contributing individuals are
properly compensated. In particular, if a sequence of queries can
indeed reveal the private data for most individuals, its price must
approach the total cost for the entire database.

The paper is organized as follows.  We describe the
basic framework for pricing private data in \autoref{sec:basic}.  In \autoref{sec:buyer}, we discuss the main required properties for price functions, developing notions of answerability for perturbed query answers and characterizing arbitrage-free price functions.  In \autoref{sec:leaks} we develop a notion of personalized privacy loss for individuals, based on differential privacy.  We define micro payment functions using this measure of privacy loss in \autoref{sec:user}.  We discuss two future challenges for pricing private data in \autoref{sec:incentives}: disclosures that could result from an individual's privacy valuations alone, and incentives for data owners to honestly reveal the valuations of their data.  We discuss related work and conclude in \autoref{sec:related} and \autoref{sec:conclusions}.


\section{Basic Concepts}

\label{sec:basic}

In this section we describe the basic architecture of the private data
pricing framework, illustrated in \autoref{fig:regular}.

\begin{figure*}[th]
	\centerline{\includegraphics[height=150pt]{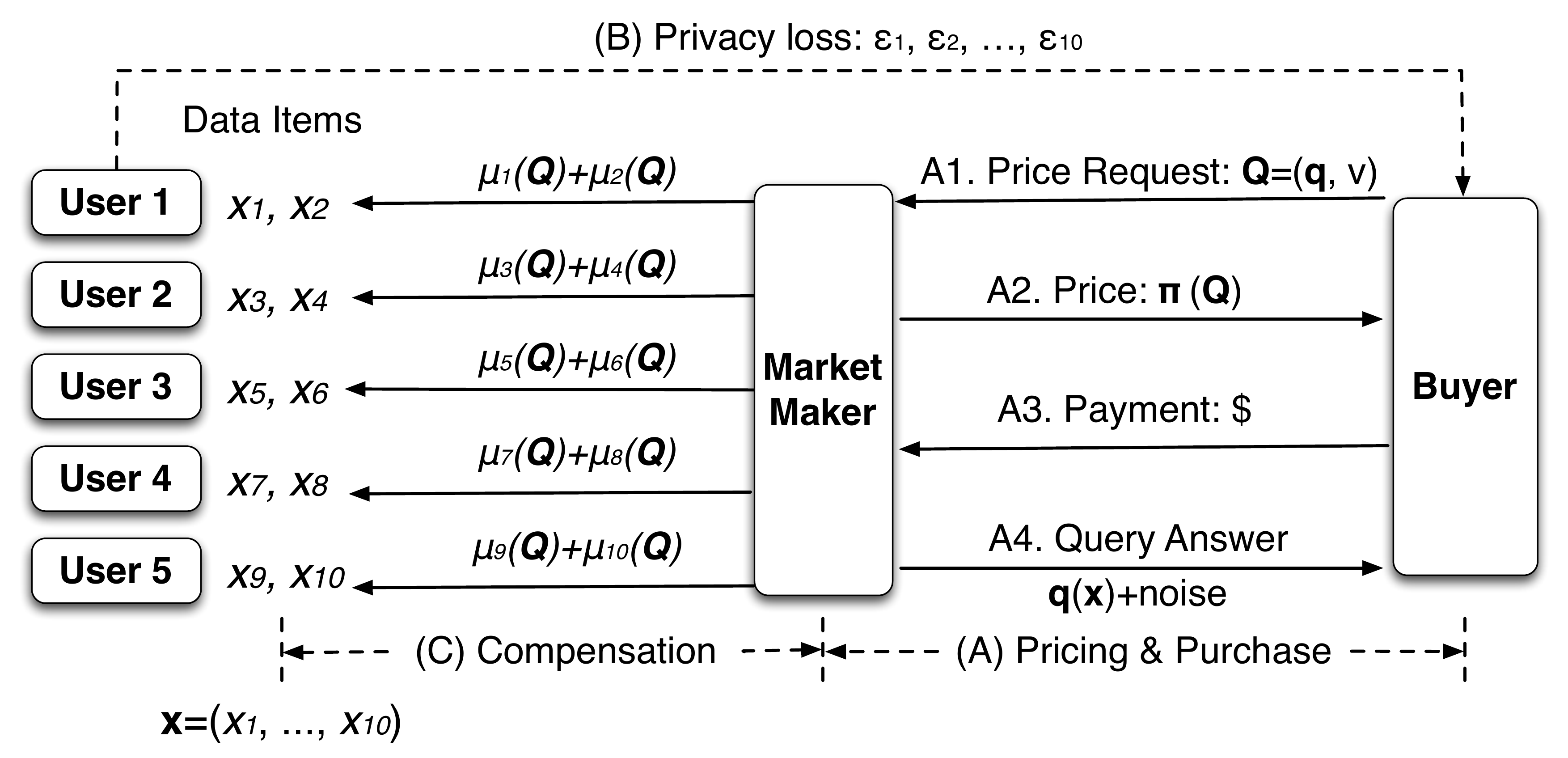}}
        \caption{\label{fig:regular} The pricing framework has three
          components: (A) Pricing and purchase: the buyer asks a query $\bbb Q = (\bbb
          q,v)$ and must pay its price, $\pi(\bbb Q)$; (B) Privacy
          loss: by answering $\bbb Q$, the market maker leaks some
          information $\varepsilon_i$ about the private data from the
          data owners to the buyer; (C) Compensation: the market maker must
          compensate each data owner for her privacy loss with
          micro-payments $\mu_i(\bbb Q)$. The pricing framework is {\em
            balanced} if the price $\pi(\bbb Q)$ is sufficient to
          cover all micro-payments $\mu_i$ and if each
          micro-payment $\mu_i$ compensates the owner for her privacy
          loss $\varepsilon_i$.}
\end{figure*}

\subsection{The Main Actors}

{\bf The Market Maker.} The market maker is trusted by the buyer and
by each of the data owners.  He collects data from the owners and
sells it in the form of queries.  When a buyer decides to purchase a
query, the market maker collects payment, computes the answer to the
query, adds noise as appropriate, returns the result to the buyer, and
finally distributes individual payments to the data owners.  The
market maker may retain a fraction of the price as profit.

{\bf The Owner and Her Data.}  Our data model is similar to that used
in \cite{Schwartz:1979:LQS:320071.320073}, where the data items are
called {\em data elements}.
\begin{definition}[Database] A database is a vector of real-valued
  data items $\bbb{x} = (x_1, x_2, \ldots, x_n)$.
\end{definition}
Each data item $x_i$ represents personal information, owned by some
individual.  In this paper we restrict the discussion to numerical
data.  For example, $x_i$ may represent an individual's rating of a
new product with a numerical value from $x_i=0$ meaning {\em poor} to
$x_i = 5$ meaning {\em excellent}; or it may represent the HIV status
of a patient in a hospital, $x_i=0$ meaning negative, and $x_i=1$
meaning positive.  Or $x_i$ may represent age, annual income, etc.
Importantly, each data item $x_i$ is owned by an individual but an
individual may own several data items. For example, if we have a table
with attributes \texttt{age}, \texttt{gender},
\texttt{marital-status}, then items $x_1, x_2, x_3$ belong to the
first individual, items $x_4, x_5, x_6$ to the second individual, etc.

\begin{sloppypar} {\bf The Buyer and His Queries.}  The buyer is a
  data analyst who wishes to compute some queries over the data.  We
  restrict our attention to the class of linear aggregation queries
  over the data items in $\bbb{x}$.
  \begin{definition}[Linear Query]
    A {\em linear query} is a real-valued vector $\bbb{q}=(q_1,q_2
    \dots q_n)$.  The answer $\bbb{q}(\bbb{x})$ to a linear query on
    $\bbb{x}$ is the vector product $\bbb{q}\bbb{x} = q_1x_1 + \dots +
    q_nx_n$. 
  \end{definition}
  Importantly, we assume that the buyer is allowed to issue multiple
  queries.  This means the buyer can combine information derived from
  multiple queries to infer answers to other queries not explicitly
  requested.  This presents a challenge we must address: to ensure
  that the buyer pays for any information that he might derive
  directly or indirectly.
\end{sloppypar}




\begin{example} \label{ex:1}
Imagine a competition between candidates $A$ and $B$ that is decided by a population of voters who each rate the competitors.  The data domain $\{0, 1, 2, 3, 4, 5\}$ represents numerical ratings.  In our data model, $x_1, x_2$ represent the rating given by Voter $1$ to candidate $A$ and $B$ respectively; $x_3, x_4$ are Voter $2$'s ratings of $A$ and $B$ respectively, and so on. The names of the voters are public, but their ratings are sensitive and should be compensated properly if used in any way.  If the buyer considers Voter $1$ and Voter $2$ experts compared with the other voters he might give a higher weight to the ratings of Voter $1$ and Voter $2$.  When a buyer wants to calculate the total rating for candidate $A$, he would issue the following linear query $\bbb{q_1} = (w_1, 0, w_1, 0, w_2, 0, w_2, 0, w_2, 0, \ldots, w_2, 0)$ with $w_1 > w_2 > 0$.
\end{example}


\subsection{Balanced Pricing Framework}

The pricing framework is {\em balanced} if (1) each data owner is
appropriately compensated whenever the answer to some query results in
some privacy loss of her data item $x_i$, and (2) the buyer is charged
sufficiently to cover all these payments.  This definition
involves three quantities: the payment $\pi$ that the buyer needs to
pay the market maker (\autoref{sec:buyer}), a measure $\varepsilon_i$
of the privacy loss of data item $x_i$ (\autoref{sec:leaks}), and a
micro-payment $\mu_i$ by which the market maker compensates the owner
of $x_i$ for this privacy loss (\autoref{sec:user}).

The buyer is allowed to specify, in addition to a linear query $\bbb
q$, an amount of noise $v$ that he is willing to tolerate in the
answer; the buyer's query is a pair $\bbb Q = (\bbb{q}, v)$, where
$\bbb{q}$ is a linear query and $v \geq 0$ represents an upper bound
on the variance. Thus, the price depends both on $\bbb q$ and $v$,
$\pi(\bbb{Q}) = \pi(\bbb{q},v) \geq 0$.  The market maker answers by
first computing the exact answer $\bbb q(\bbb x)$, then adding noise
sampled from a distribution with mean 0 and variance at most $v$.  This feature gives the buyer
more pricing options because, by increasing $v$, he can lower his
price.

Note that we define the pricing function to depend only on the
variance, and not on the type of noise used by the market maker.
However, the market participants must agree on a reasonable noise
distribution because it affects the privacy loss $\varepsilon_i$,
which further determines how much needs to be paid to the data
owners\footnote{For example, this noise $P(0)=1-2/m$, $P(\pm m) =
  1/m$, where $m=10^{64}$ (mean $0$, variance $2m$) is a poor choice.
  On one hand, it has a high variance, which implies a low price
  $\pi$.  On the other hand, it returns an accurate answer with
  extremely high probability, leading to huge privacy losses
  $\varepsilon_i$, and, consequently, to huge micro-payments.  The
  market maker will not be able to recover his costs.}.
In~\autoref{sec:leaks} we will restrict the noise to the Laplace
distribution, for which there exists an explicit formula connecting
the privacy loss $\varepsilon_i$ to the variance.

Having received the purchase price for a query $\bbb Q$, the
market-maker then distributes it to the data owners: the owner of data
item $x_i$ receives a micro payment $\mu_i(\bbb{Q}) \geq 0$.  If the
same owner contributes multiple data items $x_i, x_{i+1}, \ldots$ then
she is compensated for each.  We discuss micro-payments in \autoref{sec:user}.

Finally, the micro-payment $\mu_i(\bbb{Q})$ must compensate the data owner for
her privacy loss $\varepsilon_i$.  We say that the pricing framework defined by $\pi$, $\varepsilon_i$
and $\mu_i$ is {\em balanced} if (1) the payment received from the
buyer always covers the micro payment made to data owners, that is
$\sum_{i=1}^{n}\mu_i(\bbb Q) \leq \pi(\bbb Q)$, and (2) each
micro-payment $\mu_i$ compensates the owner of the data item $x_i$
according to the privacy loss $\varepsilon_i$, as specified by some
contract between the data owner and the market maker.  We discuss
balanced pricing frameworks and give a general procedure for designing
them in \autoref{sec:framework}.


\begin{example} \label{ex:price:perburb}
  Continuing \autoref{ex:1}, suppose that there are 1000 voters, and
  that Bob, the buyer, wants to compute the sum of ratings for
  candidate A, for which he issues the query $\bbb{q} = (1, 0, 1, 0,
  1, 0, \ldots, 1, 0)$.  Assume that each voter charges \$10 for each
  raw vote.  For an accurate answer to the query, Bob needs to pay
  $\$10,000$, which is, arguably, too expensive.  On the other hand,
  Bob could buy the query perturbed with variance $v=5,000$, which
  gives an error\footnote{$\PP{|\hat q - q| \geq 3\sqrt{2}\cdot
      \sigma} \leq 1/18 = 0.056$ (Chebyshev's inequality), where
    $\sigma = \sqrt{v} = 50\sqrt{2}$.} of $\pm 300$ with $94\%$
  confidence.  Assuming the market maker uses Laplacian noise for the
  perturbation, this query is $\varepsilon$-differentially
  private\footnote{$\varepsilon=\sqrt{2}\cdot\text{sensitivity}(\bbb{q})/\sigma=5\sqrt{2}/50\sqrt{2}=0.1$},
  with $\varepsilon=0.1$, which offers pretty good privacy to the data
  owners: each will be happy to accept only $\$0.001$ for basically no
  loss of privacy, and the buyer pays only $\$1$ for the entire query.
  The challenge is to design the prices {\em in between}.  For
  example, suppose the data owner wants to buy more accuracy, say a
  variance $v=50$ (to reduce the error to $\pm 30$), what should the
  price be now?  We will answer this in \autoref{ex:price-complete}.
  For now, let us observe that the price cannot exceed $\$100$.  If it
  did, then a savvy buyer would never pay that price, instead he would
  purchase the $\$1$ query 100 times, compute the average, and obtain
  the answer with a variance of $5000/100=50$.  This is an example of
  arbitrage and the market maker should define a pricing function that
  avoids it.
\end{example}


\section{Pricing Queries} \label{sec:buyer}

In this section we describe the first component of the framework in
\autoref{fig:regular}: the pricing function $\pi(\bbb Q) = \pi(\bbb q,
v)$.  We denote $\R^+= [0,\infty)$ and $\bar \R^+ = \R^+ \cup
\set{\infty}$.
\begin{definition} \label{def:price}
  A price function is $\pi: \R^n \times \bar \R^+\rightarrow \bar
  \R^+$.
\end{definition}
In our framework, the buyer is allowed to issue multiple queries.  As
a consequence, an important concern is that the buyer may combine
answers from multiple queries and derive an answer to a new query,
without paying the full price for the latter, a situation we call {\em
  arbitrage}.  A reasonable pricing function must guarantee that no
arbitrage is possible, in which case we call it {\em arbitrage-free}.
Such a pricing function ensures that the market maker receives proper
payment for each query by removing any incentive for the buyer to
``game'' the system by asking a set of cheaper queries in order to
obtain the desired answer.  In this section we formally define
arbitrage-free pricing functions, study their properties, and describe
a general framework for constructing arbitrage-free pricing functions,
which we will later reuse in \autoref{sec:user} to define
micro-payments, and obtain a balanced pricing framework.


\subsection{Queries and Answers}

The market maker uses a randomized mechanism for answering queries.  Given a
buyer's query $\bbb Q = (\bbb{q},v)$, the mechanism defines a random
function $\mathcal{K}_{\bbb{Q}}(\bbb x)$, such that, for any $\bbb
x$, $\EE{\mathcal{K}_{\bbb{Q}}(\bbb x)} = \bbb q(\bbb x)$ and
$\VV{\mathcal{K}_{\bbb{Q}}(\bbb x)} \leq v$.  The market maker samples one
value from this distribution and returns it to the buyer in exchange for payment $\pi(\bbb Q)$.  We abbreviate $\mathcal{K}_{\bbb{Q}}$
with $\mathcal{K}$ when $\bbb Q$ is clear from the context.
\begin{definition} \label{def:query:answer} We say that a randomized
  algorithm $\mathcal{K}(\bbb x)$ {\em answers the query $\bbb Q = (\bbb q,
    v)$ on the database $\bbb x$} if its expectation is $\bbb q(\bbb
  x)$ and its variance is less than or equal to $v$.
\end{definition}

For now, we do not impose any restrictions on the type of perturbation
used in answering the query.  The contract between the buyer and the
market maker refers only to the variance: the buyer pays for a certain
variance, and the market maker must answer with at most that
variance. The inherent assumption is that the buyer only cares about
the variance and is indifferent to other properties of the
perturbation. However, the choice of noise also affects the privacy
loss, which further affects the micro-payments: for that reason, later
in the paper (\autoref{sec:leaks}) we will restrict the perturbation
to consists of a Laplacian noise.

%

%

We assume that the market maker is stateless: he does not keep a log
of previous users, their queries, or of released answers.  As a
consequence, each query is answered using an independent random
variable.  If the same buyer issues the same query repeatedly, the
market maker answers using independent samples from the random
variable $\mathcal K$.  Of course, the buyer would have to pay for
each query separately.

\subsection{Answerability and Determinacy}

Before investigating arbitrage we establish the key concept of query
answerability.  This notion is well studied for deterministic queries
and
views~\cite{DBLP:journals/vldb/Halevy01,DBLP:journals/tods/NashSV10},
but, in our setting, the queries are random variables, and it requires
a precise definition.  Our definition below directly extends the
traditional definition from deterministic to randomized queries.


\begin{definition}[Answerability] \label{def:answerability} A query $\bbb
  Q$ is {\em answerable} from a multi-set of queries $\bbb S\!=\!
  \set{\bbb{Q_1}, \ldots, \bbb{Q_k}}$ if there exists a function $f :
  \R^k \rightarrow \R$ such that, for any mechanisms $\mathcal{K}_1$,
  $\ldots$, $\mathcal{K}_k$, that answer the queries $\bbb{Q_1},
  \ldots, \bbb{Q_k}$, the composite mechanism $f(\mathcal{K}_1, \ldots,
  \mathcal{K}_k)$ answers the query $\bbb Q$.

  We say that $\bbb Q$ is {\em linearly answerable} from $\bbb Q_1, \ldots, \bbb Q_k$ if the function $f$ is linear.
\end{definition}

For a simple example, consider queries $\bbb{Q_1}=(\bbb{q_1},v_1)$ and
$\bbb{Q_2}=(\bbb{q_2},v_2)$ and mechanisms $\mathcal{K}_1$ and $\mathcal{K}_2$ that answer them.  The query $\bbb{Q_3}=((\bbb{q_1}+\bbb{q_2})/2, (v_1+v_2)/4)$
is answerable from $\bbb{Q_1}$ and $\bbb{Q_2}$ because we can
simply sum and scale the answers returned by the two mechanisms, and
$\EE{(\mathcal{K}_1+\mathcal{K}_2)/2} = (\EE{\mathcal{K}_1} +
\EE{\mathcal{K}_2})/2$, and $\VV{(\mathcal{K}_1+\mathcal{K}_2)/2} =
(\VV{\mathcal{K}_1} + \VV{\mathcal{K}_2})/4$.  Since the function is
linear, we say that the query is linearly answerable.

%

How do we check if a query can be answered from a given set of
queries?  In this paper we give a partial answer, by
characterizing when a query is {\em linearly} answerable.


\begin{definition}[Determinacy] \label{def:det} The {\em determinacy relation} is a
  relation between a query $\bbb Q$ and a multi-set of queries $\bbb S = \set{\bbb Q_1,
    \ldots, \bbb Q_k}$, denoted $\bbb S
  \rightarrow \bbb Q$, and defined by the following rules:
  \begin{description}
  \item[Summation] \ \newline $\set{(\bbb q_1, v_1),\ldots,(\bbb q_k,
      v_k)}\rightarrow (\bbb q_1+\ldots+\bbb q_k, v_1+\ldots+v_k)$;
  \item[Scalar multiplication] $\forall c\in\mathbb{R}$, $(\bbb q,
    v)\rightarrow(c\bbb q, c^2v)$;
  \item[Relaxation] $(\bbb q, v)\rightarrow(\bbb q, v')$, where $v\leq
    v'$,
  \item[Transitivity] If\ $\bbb S_1\!\rightarrow\!\bbb Q_1, \ldots, \bbb
    S_1\!\rightarrow\!\bbb Q_k$ and $\set{\bbb Q_1,\ldots, \bbb Q_k}
    \rightarrow \bbb Q$, then $\bigcup_{i=1}^k \bbb S_k\rightarrow \bbb Q$.
  \end{description}
\end{definition}

The following proposition gives a characterization of linear answerability:

\begin{proposition}\label{prop:determinacy}
  Let $\bbb S = \set{(\bbb q_1, v_1), \ldots, (\bbb q_m, v_m)}$ be a
  multi-set of queries, and $\bbb Q = (\bbb q, v)$ be a query.  Then
  the following conditions are equivalent.
\begin{enumerate}
\item $\bbb Q$ is linearly answerable from $\bbb S$.
\item $\bbb S \rightarrow \bbb Q$.
\item  There exists $c_1, \ldots, c_m$ such that $c_1\bbb q_1+\ldots+c_m\bbb q_m=\bbb q$ and $c_1^2v_1+\ldots+c_m^2v_m\leq v$.
\end{enumerate}
\end{proposition}

\begin{proof} $(1\Leftrightarrow 3)$: Follows from the definition of linear answerability.\\
$(2\Rightarrow3)$: It is clear that in the rules of the determinacy relation, summation, scalar multiplication and relaxation are special cases of $3$. For the transitivity rule, for each $i=1, \ldots, k$, let $f_i$ be a linear function such that $f_i(\bbb S_i)=\bbb q_i$ with variance no more than $v_i$. Let $f$ be a linear function such that $f(\bbb q_1,\ldots, \bbb q_k)=\bbb q$ with variance no more than $v$. Then $f_0=f(f_1(\bbb S_1),\ldots, f_k(\bbb S_k))$ is a linear function of $\bigcup_{i=1}^k \bbb S_k$ and the variance introduced is no more than $v$.\\
$(3\Rightarrow2)$: Since $(\bbb q_i, v_i)\rightarrow(c_i\bbb q_i, c_i^2v_i)$, $\{(c_1\bbb q_1, c_1^2v_1),\ldots,$\\$(c_m\bbb q_m, c_m^2v_m)\}\rightarrow(c_1\bbb q_1+\ldots+c_m\bbb q_m, c_1^2v_1+\ldots+c_m^2v_m)=(\bbb q, c_1^2v_1+\ldots+c_m^2v_m)$ and $(\bbb q, c_1^2v_1+\ldots+c_m^2v_m)\rightarrow(\bbb q, v)$, we obtain $\bbb S\rightarrow\bbb Q$.
\end{proof}

Thus, determinacy fully characterizes linear answerability.  But it cannot characterize general answerability.  Recall that we do not specify a noise distribution in the definition of a query answering mechanism.  If the query answering mechanism does not use Gaussian noise, then non-linear composition functions may play an important role in query answering.  This follows from the existence of an unbiased non-linear estimator whose variance is smaller than linear estimators~\cite{knautz1999nonlinear} when the noise distribution is not Gaussian.


In this paper we restrict our discussion to linear answerability; in
other words, we assume that the buyer will attempt to derive new
answers from existing queries only by computing linear combinations.
By \autoref{prop:determinacy}, we will use the determinacy relation
$\bbb S \rightarrow \bbb Q$ instead of linear answerability.

Deciding determinacy, $\bbb S \rightarrow \bbb Q$, can be done in
polynomial time using a quadratic program.  The program first
determines whether $\bbb q$ can be represented as a linear combination
of queries in $\bbb S$.  If the answer is yes, the quadratic program
further checks whether there is a linear combination such that the
variance of answering $\bbb q$ with variance at most $v$.

\begin{proposition}
  Verifying whether a set $\bbb S$ of $m$ queries determines a query
  $\bbb Q$ can be done in PTIME$(m,n)$.
\end{proposition}

\begin{proof}
Given a set $\bbb S=\{(\bbb q_1, v_1), \ldots, (\bbb q_m, v_m)\}$ and a query $(\bbb q, v)$, the following quadratic program outputs the minimum possible variance to answer $\bbb q$ using linear combinations of queries in $\bbb S$.
\begin{align*}
\mbox{Given: } & \bbb q, \bbb q_1, \ldots, \bbb q_m, v_1, \ldots, v_m,\\
\mbox{Minimize: }& c_1^2v_1+\ldots+c_m^2v_m,\\
\mbox{Subject to: }& c_1\bbb q_1 + \ldots +c_m\bbb q_m = \bbb q.
\end{align*}
Once the quadratic program is solved, one can compare
$c_1^2v_1+\ldots+c_m^2v_m$ with $v$. According to the
~\autoref{prop:determinacy} $\bbb S\rightarrow (\bbb q, v)$ if and
only if $c_1^2v_1+\ldots+c_m^2v_m\leq v$. Since the quadratic program
above has $m$ variables and the constraints are a linear equation on
$n$-dimensional vectors, it can be solved in PTIME$(m,
n)$~\cite{boyd2004convex}. Thus the verification process can be done
in PTIME$(m,n)$ as well.
\end{proof}

\subsection{Arbitrage-free Price Functions: Definition}

Arbitrage is possible when the answer to a query $\bbb Q$ can be obtained more cheaply than the advertised price $\pi(\bbb Q)$ from an alternative set of priced queries.  When arbitrage is possible it complicates the interface between the buyer and market maker: the buyer may need to reason carefully about his queries to achieve the lowest price, while at the same time the market maker may not achieve the revenue intended by some of his advertised prices.

\begin{definition}[Arbitrage-free]\label{def:arbitrage}
  A price function \\$\pi(\bbb Q)$ is arbitrage-free if $\forall m \ge 1$, $\set{\bbb
    Q_1,\ldots,\bbb Q_m}\rightarrow \bbb Q$ implies: $$\pi(\bbb
  Q)\leq\sum_{i=1}^m \pi(\bbb Q_i).$$
\end{definition}

\begin{example}
  Consider a query $(\bbb q, v)$ offered for price $\pi(\bbb q, v)$.
  A buyer who wishes to improve the accuracy of the query may ask the
  same query $n$ times, $(\bbb q, v)$, $(\bbb q,v)$, $\ldots$, $(\bbb
  q,v)$, at a total cost of $n\cdot \pi(\bbb q, v)$.  The buyer then
  computes the average of the query answers to get an estimated answer
  with a much lower variance, namely $v/n$.  The price function must
  ensure that the total payment collected from the buyer covers the
  cost of this lower variance, in other words $n\cdot \pi(\bbb q,v)
  \geq \pi(\bbb q,v/n)$. If $\pi$ is arbitrage free, then it is easy
  to check that this condition holds. Indeed, $\set{(\bbb q, v),
    \ldots, (\bbb q, v)} \rightarrow (n \bbb q, nv) \rightarrow (\bbb
  q, v/n)$, and arbitrage-freeness implies $\pi(\bbb q, v/n) \leq
  \pi(\bbb q, v) + \ldots + \pi(\bbb q, v) = n \cdot \pi(\bbb q,v)$.
\end{example}

We prove that any arbitrage-free pricing function satisfies the following
simple properties:


\begin{proposition} \label{prop:arbitrage:properties}
  Let $\pi$ be an arbitrage-free pricing function.  Then:
  \begin{enumerate}[(1)]
  \item The zero query is free: $\pi(\bbb 0, v) = 0$.
  \item Higher variance is cheaper: $v \leq v'$ implies $\pi(\bbb q,
    v) \geq \pi(\bbb q, v')$.
  \item The zero-variance query is the most expensive\footnote{It is
      possible that $\pi(\bbb q, 0) = \infty$.}\!: $\pi(\bbb q,0)\!\geq
    \pi(\bbb q,v)$ for all $v \geq 0$.
  \item Infinite noise is free: if $\pi$ is a continuous function,
    then $\pi(\bbb q, \infty) = 0$.
  \end{enumerate}
\end{proposition}

\begin{proof}
  For (1), we have $\emptyset \rightarrow (\bbb 0, 0)$ by
  the first rule of \autoref{def:det} (taking $k=0$, i.e.  $\bbb S =
  \emptyset$) and $(\bbb 0, 0) \rightarrow (\bbb 0, v)$ by the third
  rule; hence $\pi(\bbb 0, v) = 0$.  (2) follows from
  $(\bbb q, v) \rightarrow (\bbb q, v')$ when $v \leq v'$.  (3) follows immediately, since all variances are $v \geq 0$.  For (4), we use the second rule to derive $(1/c\cdot \bbb q,
  v) \rightarrow (\bbb q, c^2 \cdot v)$, hence $\pi(\bbb q, \infty) =
  \lim_{c \rightarrow \infty} \pi(\bbb q, c^2 \cdot v) \leq \lim_{c
    \rightarrow \infty} \pi(1/c \cdot \bbb q, v) = \pi(\bbb 0, v) = 0$.
\end{proof}

Arbitrage-free price functions have been studied before
\cite{DBLP:conf/pods/KoutrisUBHS12,lipricing}, but only in the context
of deterministic (i.e. unperturbed) query answers.  Our definition
extends those in \cite{DBLP:conf/pods/KoutrisUBHS12,lipricing} to
queries with perturbed answers.

\subsection{Arbitrage-free Price Functions: Synthesis}

Next we address the question of how to design arbitrage-free
pricing functions.  Obviously, the trivial pricing function $\pi(\bbb Q) = 0$,
for all $\bbb Q$, under which every query is free, is
arbitrage-free, but we want to design non-trivial pricing functions.
For example, it would be a mistake for the market-maker to charge a
constant price $c > 0$ for each query, i.e. $\pi(\bbb Q) = c$ for all
$\bbb Q$, because such a pricing function leads to arbitrage (this
follows from \autoref{prop:arbitrage:properties}).

We start by analyzing how an arbitrage-free price function $\pi(\bbb q, v)$
depends on the variance $v$.  By (2) of \autoref{prop:arbitrage:properties} we know that it is monotonically
decreasing in $v$, and by (4) it cannot be independent of
$v$ (unless $\pi$ is trivial).  The next proposition shows that it
cannot decrease faster than $1/v$:
%
\begin{proposition}\label{prop:vardecay}
  For any arbitrage-free price function $\pi$ and any linear query
  $\bbb q$, $\pi(\bbb q, v)=\Omega(1/v)$.
\end{proposition}
\begin{proof}
  Suppose the contrary: there exists a linear query $\bbb q$ and a
  sequence $\{v_i\}_{i=1}^\infty$ such that
  $\lim_{i\rightarrow\infty}v_i=+\infty$ and
  $\lim_{i\rightarrow\infty}v_i\pi(\bbb q, v_i)=0$. Select $i_0$ such
  that $v_{i_0}>1$ and $v_{i_0}\pi(\bbb q, v_{i_0}) < \pi(\bbb q,
  1)/2$. Then, we can answer $\pi(\bbb q, 1)$ by asking the query
  $\pi(\bbb q, v_{i_0})$ at most $\lceil v_{i_0}\rceil$ times and
  computing the average.
  For these  $\lceil v_{i_0}\rceil$ queries we pay:
$$\lceil v_{i_0}\rceil\pi(\bbb q, v_{i_0})\leq (v_{i_0}+1)\pi(\bbb q, v_{i_0})<2v_{i_0}\pi(\bbb q, v_{i_0})<\pi(\bbb q, 1),$$
which implies that we have arbitrage, a contradiction.
\end{proof}
%

Our next step is to understand the dependency on $\bbb q$, and for
that we will assume that $\pi$ is inverse proportional to $v$, in
other words that it decreases at a rate $1/v$, which is the fastest
rate allowed by the previous proposition.  Set $\pi(\bbb q, v) =
f^2(\bbb q)/v$, for some positive function $f$ that depends only on
$\bbb q$.  We prove that $\pi$ is arbitrage-free iff $f$ is a
semi-norm.  Recall that a {\em semi-norm} is a function $f:
\mathbb{R}^n \rightarrow \mathbb{R}$ that satisfies the following
properties\footnote{Taking $c=0$ in the first property implies $f(\bbb
  0) = 0$; if the converse also holds, i.e.  $f(\bbb q) = 0$ implies
  $\bbb q = \bbb 0$, then $f$ is called a {\em norm}.  Also, recall
  that any semi-norm satisfies $f(\bbb q) \geq 0$, by the triangle
  inequality.}:
\begin{itemize}
\item For any $c \in \mathbb{R}$ and any $\mathbf{q} \in \mathbb{R}^n$, $f(c\mathbf{q}) = |c|f(\mathbf{q})$.
\item For any $\mathbf{q}_1$, $\mathbf{q}_2 \in \mathbb{R}^n$,
  $f(\mathbf{q}_1 + \mathbf{q}_2) \le f(\mathbf{q}_1) +
  f(\mathbf{q}_2)$.
\end{itemize}

%

We prove:

\begin{theorem}\label{prop:abfreetonorm}
  Let $\pi(\bbb q, v)$ be a price function s.t.
  $\pi(\bbb q,v) = f^2(\bbb q)/v$ for some function $f$.\footnote{In other
    words, $f(\bbb q)=\sqrt{\pi(\bbb q,v)v}$ is independent of $v$.}  Then
  $\pi(\bbb q, v)$ is arbitrage-free iff $f(\bbb q)$ is a semi-norm.
\end{theorem}

\begin{proof}
  $(\Rightarrow):$ Assuming $\pi$ is arbitrage-free, we prove that $f$
  is a semi-norm.  For $c \ne 0$, by the second rule of
  \autoref{def:det}, we have both:
\begin{align*}
(\bbb q, v)\rightarrow & (c\bbb q, c^2v) \\
(c\bbb q, c^2v)  \rightarrow & (\frac{1}{c}\times c\bbb q,
(\frac{1}{c})^2 \times c^2v) \rightarrow (\bbb q, v)
\end{align*}
Therefore both $\pi(\bbb q, v) \le \pi(c\bbb q, c^2v)$ and $\pi(\bbb
q, v) \ge \pi(c\bbb q, c^2v)$ hold, thus $\pi(\bbb q, v) = \pi(c\bbb
q, c^2v)$.  This implies that, if $c \neq 0$,
$$f(c\bbb q)=\sqrt{\pi(c\bbb q, c^2v)c^2v}=|c|\sqrt{\pi(\bbb q, v)v}= |c|f(\bbb q).$$
If $c = 0$, we also have $f(c\bbb q)=\sqrt{\pi(c\bbb q, c^2v)c^2v} = 0
= |c|f(\bbb q)$.

Next we prove that $f(\bbb q_1 + \bbb q_2) \leq f(\bbb q_1) + f(\bbb
q_2)$.  Set the variances $v_1 = f(\bbb q_1)$ and $v_2 = f(\bbb q_2)$;
then we have $f(\bbb q_1) = \pi(\bbb q_1,v_1)$ and $f(\bbb q_2) =
\pi(\bbb q_2,v_2)$.  By the first rule in \autoref{def:det} we have
$\set{(\bbb q_1,v_1), (\bbb q_2, v_2)} \rightarrow (\bbb q_1+\bbb q_2,
v_1+v_2)$, and therefore:
\begin{align*}
  \frac{f^2(\bbb q_1+\bbb q_2)}{f(\bbb q_1) + f(\bbb q_2)} = & \pi(\bbb
  q_1 + \bbb q_2, v_1+v_2) \\
  \leq & \pi(\bbb q_1, v_1) + \pi(\bbb q_2, v_2) = f(\bbb q_1) +
  f(\bbb q_2)
\end{align*}
which proves the claim.


$(\Leftarrow):$ Suppose $\pi(\bbb q,v)=f^2(\bbb q)/v$ and $f(\bbb q)$
is a semi-norm. According to~\autoref{prop:determinacy}, $\{(\bbb q_1,
v_1),\ldots,(\bbb q_m, v_m)\}\!\rightarrow\!(\bbb q, v)$ if and only
if there exists $c_1, \ldots, c_m$ such that $c_1\bbb
q_1+\ldots+c_m\bbb q_m=\bbb q$ and $c_1^2v_1+\ldots+c_m^2v_m\leq
v$. Then,
\begin{align*}
\sum_{i=1}^m\pi(\bbb q_i, v_i)&=\sum_{i=1}^m\frac{f^2(\bbb q_i)}{v_i}=\frac{(\sum_{i=1}^m\frac{f^2(\bbb q_i)}{v_i})(\sum_{i=1}^m c_i^2v_i)}{\sum_{i=1}^m c_i^2v_i}\\
&\geq \frac{(\sum_{i=1}^m |c_i|f(\bbb q_i))^2}{\sum_{i=1}^m c_i^2v_i}=\frac{(\sum_{i=1}^m f(c_i\bbb q_i))^2}{\sum_{i=1}^m c_i^2v_i}\\
&\geq \frac{f(\bbb q)^2}{v}=\pi(\bbb q, v),
\end{align*}
where the first inequality follows from the Cauchy-Schwarz inequality and the second comes from the sub-additivity of the semi-norm.
\end{proof}



As an immediate application of the theorem, let us instantiate $f$
to be one of the norms $L_2, L_\infty, L_p$, or a weighted $L_2$ norm.
This implies that the following four functions are arbitrage-free:
\begin{align}
  \pi(\bbb q, v) = & ||\bbb q||_2^2 / v = \sum_i q_i^2 / v  && \label{eq:l2} \\
  \pi(\bbb q, v) = & ||\bbb q||_\infty^2 / v = \max_i q_i^2 / v \label{eq:linfty} \\
  \pi(\bbb q, v) = & ||\bbb q||_p^2 / v = (\sum_i q_i^p)^{2/p} / v &&  p \geq 1 \label{eq:lp} \\
  \pi(\bbb q, v) = & (\sum_i w_i \cdot q_i^2) / v && w_1, \ldots, w_n  \geq 0 \label{eq:weighted:norm}
\end{align}

However, these are not the only arbitrage-free pricing functions: the
proposition below gives us a general method for synthesizing new
arbitrage-free pricing functions from existing ones.  Recall that a
function $f : (\bar \R^+)^k \rightarrow \bar \R^+$ is called {\em
  subadditive} if for any two vectors $\bbb x, \bbb y \in (\bar
\R^+)^k$, $f(\bbb x + \bbb y) \leq f(\bbb x) + f(\bbb y)$; the
function is called {\em non-decreasing} if $\bbb x \leq \bbb y$
implies $f(\bbb x) \leq f(\bbb y)$.

\begin{proposition} \label{prop:composition} Let $f :  (\bar \R^+)^k
  \rightarrow \bar \R^+$ be a subadditive, non-decreasing function.
  For any arbitrage-free price functions $\pi_1, \ldots, \pi_k$,
  the function $\pi(\bbb Q) = f(\pi_1(\bbb Q), \ldots, \pi_k(\bbb Q))$
  is also arbitrage-free.
\end{proposition}
\begin{proof}
  For any query $\bbb Q$, let $\bar \pi(\bbb Q) = (\pi_1(\bbb Q),
  \ldots, \pi_k(\bbb Q))$.  Assume $\{(\bbb q_1, v_1),\ldots,(\bbb
  q_m, v_m)\}\rightarrow(\bbb q, v)$.  We have:
  \begin{align*}
    \bar \pi(\bbb Q) \leq & \sum_i \bar \pi(\bbb Q_i) && \mbox{because  each $\pi_j$ is arbitrage-free}\\
   f(\bar \pi(\bbb Q)) \leq & f(\sum_i \bar \pi(\bbb Q_i)) && \mbox{because $f$ is non-decreasing}\\
                       \leq & \sum_i f(\bar \pi(\bbb Q_i))  && \mbox{because $f$ is sub-additive}
  \end{align*}
\end{proof}
\autoref{prop:composition} allows us to synthesize new arbitrage-free
price function from existing arbitrage-free price functions. Below we
include some operations that satisfy the requirements
in~\autoref{prop:composition}.
\begin{corollary}\label{cor:compfunction}
If $\pi_1, \ldots, \pi_k$ are arbitrage-free price functions, then so are the following functions:
\begin{itemize}
\item{\em Linear combination:} $c_1\pi_1+\ldots+c_k\pi_k$, $c_1,
  \ldots,c_k \geq 0$.
\item{\em Maximum:} $\max(\pi_1, \ldots, \pi_k)$;
\item{\em Cut-off:} $\min(\pi_1, c)$, where $c\geq 0$;
\item{\em Power:} $\pi_1^c$ where $0<c\leq 1$;
\item{\em Logarithmic:} $\log(\pi_1+1)$;
\item{\em Geometric mean:} $\sqrt{\pi_1 \cdot \pi_2}$.
\end{itemize}
\end{corollary}
\begin{proof} It is clear that all the functions above are
  monotonically increasing.  One can check directly that maximum and
  cut-off functions are sub-additive.  Sub-additivity for the rest
  follows from the following:

  \begin{lemma} \label{lemma:concave}
    Let $f : (\bar \R^+)^k \rightarrow \bar \R^+$ be a non-decreasing
    function s.t.  $f(\bbb 0) = 0$ and all second derivatives are
    continuous.  Then, if $\partial^2 f / \partial x_i \partial x_j
    \leq 0$ for all $i, j = 1,\ldots, k$, then $f$ is sub-additive.
  \end{lemma}

  \begin{proof}
    Denote $f_i= \partial f/\partial x_i$ and $f_{ij} = \partial^2
    f/ \partial x_i \partial x_j$.  We apply twice the first-order
    Taylor approximation $f(\bbb x) - f(\bbb 0) = \sum_i (\partial
    f/\partial x_i)(\bbb \xi) \cdot x_i$, once to $g(\bbb y) = f(\bbb
    x + \bbb y) - f(\bbb y)$, and the second time to $h(\bbb x) =
    \sum_j (f_j(\bbb x+\bbb \xi) - f_j(\xi)) \cdot y_j$:
    \begin{align*}
      & f(\bbb x) + f(\bbb y) - f(\bbb x + \bbb y) =  [f(\bbb x) -  f(\bbb 0)] + [f(\bbb x + \bbb y)  - f(\bbb y)]  \\
      & = g(\bbb 0) - g(\bbb y)  = - \sum_j g_j(\bbb \xi) \cdot y_j \\
      & = - \sum_j (f_j(\bbb x+\bbb \xi) - f_j(\xi)) \cdot y_j  = -
      \sum_{ij} f_{ij}(\bbb \eta + \bbb \xi) \cdot x_i \cdot y_j \geq 0
    \end{align*}
  \end{proof}

\end{proof}
%

\begin{example}
  For a simple illustration we will prove that the pricing function
  $\pi(\bbb q, v) = \max_i |q_i| / \sqrt{v}$ is arbitrage free.  Start
  from $\pi_1(\bbb q, v) = \max_i q_i^2 / v$, which is arbitrage-free
  by \autoref{eq:linfty}, then notice that $\pi = (\pi_1)^{1/2}$,
  hence $\pi$ is arbitrage-free by \autoref{cor:compfunction}.
\end{example}

\subsection{Selling the True Private Data}

\label{subsec:true:data}

While under differential privacy perturbation is always necessary, in
data markets the data being sold is usually unperturbed.  Perturbation
is only a tool to reduce the price for the buyer.  Therefore, a
reasonable pricing function $\pi(\bbb{q},v)$ needs to give a finite
price for a zero variance, and none of our simple pricing functions in
\autoref{eq:l2}-\autoref{eq:weighted:norm} have this property.

One can design arbitrage-free pricing functions that return a finite
price for the unperturbed data by using any bounded function with the
properties required by~\autoref{prop:composition}.  For example, apply
the cut-off function (~\autoref{cor:compfunction}) to any of the
pricing functions in \autoref{eq:l2}-\autoref{eq:weighted:norm}.  More
sophisticated functions are possible by using sigmoid curves, often
used as learning curves by the machine learning community. Many of
those curves are concave and monotonically increasing over
$\mathbb{R}^+$, which, by \autoref{lemma:concave}, are subadditive on
$\R^+$ when $f(0)=0$.  Thus, we can apply functions of those learning
curves that are centered at $0$ to~\autoref{prop:composition} so as to
generate smooth arbitrage-free price functions with finite maximum.
Other such functions are given by the following (the proof in the appendix):

\begin{corollary}\label{COR:FINITE}
  Given an arbitrage-free price function $\pi$, each of the following
  functions is also arbitrage-free and bounded: $\texttt{atan}(\pi)$,
  $\texttt{tanh}(\pi)$, $\pi/\sqrt{\pi^2+1}$.
\end{corollary}

\begin{example} \label{ex:price-complete} Suppose we want to charge a
  price $p$ for the true, unperturbed result of a query $\bbb q$.
  Assume $||\bbb q||_2^2 = n$, and let $\pi_1(\bbb{q}, v) = ||\bbb
  q||_2^2 / v = n/v$ be the pricing function in \autoref{eq:l2}.  It
  follows that the function\footnote{We use $\Pi$ for the constant
    {\em pi} to avoid confusion with the pricing function $\pi$.}
  \begin{align*}
    \pi(\bbb q, v) = & \frac{2p}{\Pi} \cdot \texttt{atan}(c \cdot
    \pi_1(\bbb q, v)) = \frac{2p}{\Pi} \cdot \texttt{atan}(c\frac{n}{v})
  \end{align*}
  is arbitrage-free.  Here $c>0$ is a parameter.  For example, suppose
  the buyer cannot afford the unperturbed query ($v=0$), and settles
  instead for a variance $v = \Theta(n)$ (it corresponds to a standard
  deviation $\sqrt{n}$, which is sufficient for some applications);
  for concreteness, assume $v = 5n$.  Then $\pi(\bbb q, v) =
  \frac{2p}{\Pi} \cdot \texttt{atan}(c/5)$.  To make this price
  affordable, we choose $c \ll 1$, in which case the price becomes
  $\pi \approx 2\cdot c\cdot p/(5\cdot \Pi) = 0.13 \cdot c\cdot p$.
  In \autoref{ex:price:perburb} the price of the unperturbed query was
  $p=\$10,000$, and we wanted to charge $\$1$ for the variance $v = 5n
  = 5000$: for that we can use the pricing function $\pi$ above, with
  $c = 1/ (0.13\cdot p) = 7.85\cdot 10^{-4}$.  We can now answer the
  question in \autoref{ex:price:perburb}: the cost of the query with
  variance $v=50$ is $\pi(\bbb q, v) = \frac{2p}{\Pi} \cdot
  \texttt{atan}(100\cdot c/5)=\$99.94$.
%
\end{example}

\section{Privacy Loss}

\label{sec:leaks}

In this section we describe the second component of the pricing
framework in \autoref{fig:regular}: the privacy loss $\varepsilon_i$.
Recall that, for each buyer's query $\bbb Q = (\bbb{q},v)$, the market
maker defines a random function $\mathcal{K}_{\bbb{Q}}$, such that,
for any database instance $\bbb x$, the random variable
$\mathcal{K}_{\bbb{Q}}(\bbb x)$ has expectation $\bbb q(\bbb x)$ and
variance less than or equal to $v$.  By answering the query through this
mechanism, the market maker leaks some information about each data
item $x_i$, and its owner expects to be compensated appropriately.  In
this section we define formally the privacy loss, and establish a few
of its properties. In the next section we will relate the privacy loss to the
micro-payment that the owner expects.

Our definition of privacy loss is adapted from differential privacy,
which compares the output of a mechanism with and without the
contribution of the data item $x_i$.  For that, we need to impose a
bound on the possible values of $x_i$.  We fix a bounded domain of
values $X \subseteq \R$, and assume that each data item $x_i$ is in
$X$.  For example, in case of binary data values $X = \set{0,1}$ ($0=$
owner does not have the feature, $1=$ she does have the feature), or
in case of ages, $X = [0, 150]$, etc.

Given the database instance $\mathbf{x}$, denote by $\mathbf{x}^{(i)}$
the database instance obtained by setting $x_i = 0$ and leaving all
other values unchanged.  That is, $\mathbf{x}^{(i)}$ represents the
database without the item $i$.

\begin{definition} \label{def:privacy:loss} Let $\mathcal{K}$ be any
  mechanism (meaning: for any database instance $\bbb x$,
  $\mathcal{K}(\bbb x)$ is a random variable).  The {\em privacy loss}
  to user $i$, in notation $\varepsilon_i(\mathcal{K}) \in \bar \R^+$
  is defined as:
  \begin{align*}
    \varepsilon_i(\mathcal{K}) = & \text{sup}_{S,\bbb x} \left| \log
      \frac{\PP{\mathcal{K}(\bbb x) \in S}}{\PP{\mathcal{K}(\bbb
          x^{(i)}) \in S}} \right|
  \end{align*}
  \noindent where $\bbb x$ ranges over $X^n$ and $S$ ranges over
  measurable sets of $\R$.
\end{definition}

We explain the connection to differential privacy in the next section.
For now, we derive some simple properties of the privacy loss
function.  The following are well
known~\cite{DBLP:conf/tcc/DworkMNS06}:

\begin{proposition} (1) Suppose $\mathcal{K}$ is a deterministic
  mechanism.  Then $\varepsilon_i(\mathcal{K})=0$ when $\mathcal{K}$
  is independent of the input $x_i$, and
  $\varepsilon_i(\mathcal{K})=\infty$ otherwise.  (2) Let
  $\mathcal{K}_1, \ldots, \mathcal{K}_m$, be mechanisms with privacy
  losses $\varepsilon_1, \ldots, \varepsilon_m$.  Let $\mathcal{K} =
  c_1 \cdot \mathcal{K}_1 + \ldots + c_m \cdot \mathcal{K}_m$ be a new mechanism 
computed using a linear combination.  Then its privacy loss is
  $\varepsilon(\mathcal{K}) = |c_1| \cdot \varepsilon_1 + \ldots +
  |c_m| \cdot \varepsilon_m$.
\end{proposition}


In this paper we restrict the mechanism to be data-\\independent.

\begin{definition}
  A query-answering mechanism $\mathcal{K}$ is called {\em data
    independent} if, for any query $\bbb Q = (\bbb q, v)$,
  $\mathcal{K}_{\bbb Q}(\bbb x) = \bbb{q}(\bbb x) + \rho(v)$, where
  $\rho(v)$ is a random function.
\end{definition}

In other words, a data-independent mechanism for answering $\bbb Q =
(\bbb q, v)$ will first compute the true query answer $\bbb q(\bbb
x)$, then add a noise $\rho(v)$ that depends only on the buyer's
specified variance, and is independent on the database instance.  We
prove:

\begin{proposition}
  Let $\mathcal{K}$ be any data-independent mechanism.  If the query
  $\bbb Q = (\bbb q, v)$ has the $i^{th}$ component equal to zero,
  $q_i = 0$, then $\varepsilon_i(\mathcal{K}_{\bbb Q}) = 0$.  In other
  words, users who do not contribute to a query's answer suffer no
  privacy loss.
\end{proposition}

\begin{proof}
  The two random variables $\mathcal{K}_{\bbb Q}(\bbb x)$ and
  $\mathcal{K}_{\bbb Q}(\bbb x^{(i)})$ are equal, because
  $\mathcal{K}_{\bbb Q}(\bbb x) = \bbb{q}(\bbb x) + \rho(v) =
  \bbb{q}(\bbb x^{(i)}) + \rho(v) = \mathcal{K}_{\bbb Q}(\bbb
  x^{(i)})$, which proves the claim.
\end{proof}

In contrast, a data-dependent mechanism might compute the noise as a
function of all data items $\bbb x$, and may result in a privacy loss
for the data item $x_i$ even when $q_i= 0$.  For that reason we only
consider data-independent mechanisms in this paper.

The privacy loss given by \autoref{def:privacy:loss} is difficult to
compute in general.  Instead, we will follow the techniques developed
for differential privacy, and give an upper bound based on query
sensitivity.  
Let $\gamma = \text{sup}_{x \in X}|x|$.

\begin{definition}[Personalized Sensitivity]\label{def:person:sensitivity}
 \hspace*{-5pt}The \emph{sensitivity} $s_i$ of a query $\bbb q$ at data item $x_i$ is
  defined as
\begin{align*}
  s_i = \text{sup}_{\bbb{x}\in X^n} |\bbb q(\bbb{x}) - \bbb q(\bbb{x}^{(i)})| = \gamma \cdot |q_i|.
\end{align*}
\end{definition}

We let $Lap(b)$ denote the one-dimensional Laplacian distribution
centered at $0$ with scale $b$ and the corresponding probability
density function $g(x) = \frac{1}{2\cdot b}e^{-\frac{|x|}{b}}$.

\begin{definition} \label{def:laplace:mechanism} The {\em Laplacian
    Mechanism}, denoted $\mathcal{L}$, is the data-independent
  mechanism defined as follows: for a given query $\bbb Q = (\bbb q,
  v)$ and database instance $\bbb x$, the mechanism returns
  $\mathcal{L}_{\bbb Q}(\bbb x) = \bbb q(\bbb x) + \rho$, where $\rho$
  is noise with distribution $Lap(b)$ and $b = \sqrt{v/2}$.
\end{definition}

The following is known from the work on differential
privacy~\cite{DBLP:conf/tcc/DworkMNS06}.

\begin{proposition}\label{prop:laplace}
  Let $\mathcal{L}$ be the Laplacian mechanism and $\bbb Q = (\bbb q,
  v)$ be a query.  Then, the privacy loss of individual $i$ is bounded
  by:
\begin{align*}
  \varepsilon_i(\mathcal{L}_{\bbb Q}) \le \frac{\gamma}{\sqrt{v/2}}  |q_i|.
\end{align*}
\end{proposition}



\section{Micro-Payments to Data Owners}

\label{sec:user}

In this section we describe the third component of the pricing
framework of \autoref{fig:regular}: the micro-payments $\mu_i$.  By
answering a buyer's query $\bbb Q$, using some mechanism
$\mathcal{K}_{\bbb Q}$, the market maker leaks some of the private
data of the data owners; he must compensate each data owner with a
micro-payment $\mu_i(\bbb Q)$, for each data item $x_i$ that they
own. The micro-payment close the loop in \autoref{fig:regular}: they
must be covered by the buyer's payment $\pi$, and must also be a
function of the degree of the privacy loss $\varepsilon_i$.  We make
these connections precise in the next section.  Here, we state two
simple properties that we require the micro-payments to satisfy.


\begin{definition} \label{def:micro:d1} Let $\mu_i$ be a
  micro-payment function.  We define the following two properties:
  \begin{description}
  \item[Fairness] For each $i$, if $q_i = 0$, then $\mu_i(\mathbf{q},
    v) = 0$.
  \item[Micro arbitrage-free] For each $i$, $\mu_i(\bbb Q)$ is an
    arbitrage-free pricing function.
  \end{description}
\end{definition}

Fairness is self-explanatory: data owners whose data is not queried
should not expect payment.  Arbitrage-freeness is a promise that the
owner's loss of privacy will be compensated, and that there is no way
for the buyer to circumvent the due micro-payment by asking other
queries and combining their answers.  This is similar to, but distinct
from arbitrage-freeness of $\pi$, and must be verified for each user.

\section{Balanced Pricing Frameworks}

\label{sec:framework}

Finally, we discuss the interaction between the three components in
\autoref{fig:regular}, the query price $\pi$, the privacy loss
$\varepsilon_i$, and the micro-payments $\mu_i$, and define formally
when a pricing framework is {\em balanced}.  Then, we give a general
procedure for designing a balanced pricing framework.

\subsection{Balanced Pricing Frameworks: Definition}


The contract between the data owner of item $x_i$ and the market-maker
consists of a non-decreasing function $W_i : \bar \R^+ \rightarrow
\bar \R^+$, s.t. $W_i(0)=0$.  This function represents a guarantee to
the data owner that she will be compensated with at least $\mu_i \geq
W_i(\varepsilon_i)$ in the event of a privacy loss $\varepsilon_i$.
We denote $\bbb W = (W_1, \ldots, W_n)$ the set of contracts between the
market-maker and all data owners.

The connection between the micro-payments $\mu_i$, the query price
$\pi$ and the privacy loss $\varepsilon_i$ is captured by the
following definition.
\begin{definition} \label{def:micro:d2} We say that the micro-payment
  functions $\mu_i$, $i=1,\ldots, n$ are {\em cost-recovering} for a
  pricing function $\pi$ if, for any query $\bbb Q$, $\pi(\bbb Q) \ge
  \sum_i \mu_i(\bbb Q)$.

  Fix a query answering mechanism $\mathcal{K}$.  We say that a
  micro-payment function $\mu_i$ is {\em compensating} for a contract
  function $W_i$, if for any query $\bbb Q$, $\mu_i(\bbb Q) \geq
  W_i(\varepsilon_i(\mathcal{K}_{\bbb Q}))$.
\end{definition}

The market maker will insist that the micro-payment functions is
cost-recovering: otherwise, he will not be able to pay the data owners
from the buyer's payment.  A data owner will insist that the
micro-payment function is compensating: this enforces the contract
between her and the market-maker, guaranteeing that she will be
compensated at least $W_i(\varepsilon_i)$, in the event of a privacy
loss $\varepsilon_i$.

Fix a query answering mechanism $\mathcal{K}$.  We denote a pricing
framework $(\pi, \bbb{\bf \varepsilon}, \bbb \mu, \bbb W)$, where
$\pi(\bbb Q)$, $\mu_i(\bbb Q)$ are the buyer's price and the
micro-payments, $\varepsilon = (\varepsilon_1, \ldots, \varepsilon_n)$
where $\varepsilon_i(\mathcal{K}_{\bbb Q})$ is the privacy loss
corresponding to the mechanism $\mathcal{K}$, and $W_i(\varepsilon)$
is the contract with the data owner $i$.


\begin{definition}
  A pricing framework $(\pi, \bbb \varepsilon, \bbb \mu, \bbb W)$
  is {\em balanced} if (1) $\pi$ is arbitrage-free and (2) the
  micro-payment functions $\bbb \mu$ are fair, micro arbitrage-free,
  cost-recovering for $\pi$, and compensating for $\bbb W$.
\end{definition}


We explain how the contract between the data owner and the market
maker differs from that in privacy-preserving mechanisms.  Let
$\varepsilon > 0$ be a small constant.  A mechanism $\mathcal{K}$ is
called {\em differentially private}~\cite{DBLP:journals/cacm/Dwork11}
if, for any user $i$ and for any measurable set $S$, and any database
instance $\bbb x$:
\begin{align*}
  \PP{\mathcal{K}(\bbb x) \in S} \le e^{\varepsilon} \times
  \PP{\mathcal{K}(\bbb{x}^{(i)}) \in S}
\end{align*}
In differential privacy, the basic contract between the mechanism
and the data owner is the promise to every user that her privacy loss
is no larger than $\varepsilon$.  In our framework for pricing private
data we turn this contract around.  Now, privacy {\em is} lost, and
\autoref{def:privacy:loss} quantifies this loss.  The contract is that
the users are compensated according to their privacy loss.  At an
extreme, if the mechanism is $\varepsilon$-differentially private for
a tiny $\varepsilon$, then each user will receive only a tiny
micro-payment $W_i(\varepsilon)$; as her privacy loss increases, she
will be compensated more.

The micro-payments circumvent a fundamental limitation of
differentially-private mechanisms.  In differential privacy, the buyer
has a fixed budget $\varepsilon$ for all queries that he may ever ask.
In order to issue $N$ queries, he needs to divide the privacy budget
among these queries, and, as a result, each query will be perturbed
with a higher noise; after issuing these $N$ queries, he can no longer
query the database, because otherwise the contract with the data owner
would be breached.  In our pricing framework there is no such
limitation, because the buyer simply pays for each query.  The budget
is now a real dollar budget, and the buyer can ask as many query as he
wants, with as high accuracy as he wants, as long as he has money to
pay for them.

\subsection{Balanced Pricing Frameworks: Synthesis}

Call $(\bbb \varepsilon, \bbb \mu, \bbb W)$ {\em semi-balanced} if all
micro-payment functions are fair, micro-arbitrage free, and
compensating w.r.t. $\mathcal{K}$; that is, we leave out the pricing
function $\pi$ and the cost-recovering requirement.  The first step is to design a semi-balanced  set of micro-payment
functions.

\begin{proposition}\label{prop:basic}
  Let $\mathcal{L}$ be the Laplacian Mechanism, and let the contract
  functions be linear, $W_i(\varepsilon_i) = c_i \cdot \varepsilon_i$,
  where $c_i > 0$ is a fixed constant, for $i=1,\ldots,n$.  Define the
  micro-payment functions $\mu_i(\bbb Q) = \frac{\gamma \cdot
    c_i}{\sqrt{v/2}}|q_i|$, for $i=1,\ldots,n$.  Then $(\bbb \varepsilon,
  \bbb \mu, \bbb W)$ is semi-balanced.
\end{proposition}

\begin{proof}
  Each $\mu_i$ is fair, because $q_i = 0$ implies $\mu_i = 0$.  By
  setting $w_i = 2 \gamma^2 \cdot c_i^2$ and $w_j = 0$ for $j \neq
  i$ in \autoref{eq:weighted:norm}, the function $\pi_i(\bbb Q) =
  \frac{2 \gamma^2 \cdot c_i^2 \cdot q_i^2}{v}$ is arbitrage free.
  By \autoref{cor:compfunction}, the function $\mu_i(\bbb Q) =
  \left(\pi_i(\bbb Q)\right)^{1/2}$ is also arbitrage-free, which
  means that $\mu_i$ is micro-arbitrage free.  Finally, by
  \autoref{prop:laplace}, we have $W_i(\varepsilon_i(\mathcal{L}_{\bbb
    Q})) = c_i \cdot \varepsilon_i(\mathcal{L}_{\bbb Q})  \le c_i
  \frac{\gamma}{\sqrt{v/2}}  |q_i| = \mu_i(\bbb Q)$, proving that
  $\mu_i$ is compensating.
\end{proof}

Next, we show how to derive new semi-balanced micro-payments from
existing ones.

\begin{proposition}\label{prop:f:mu} Suppose that $(\bbb \varepsilon,
  \bbb \mu^j, \bbb W^j)$ is semi-\\balanced, for $j=1,\ldots,k$ (where $\bbb
  \mu^j = (\mu_1^j, \ldots, \mu_n^j)$, and $\bbb W^j = (W^j_1, \ldots,
  W^j_n)$, for $j=1,\ldots,k$), and let $f_i : (\bar \R^+)^k \rightarrow \bar
  \R^+$, $i=1,\ldots,n$, be $n$ non-decreasing, sub-additive functions
  s.t. $f_i(\bbb 0) = 0$, for all $i=1,\ldots,n$.  Define $\mu_i =
  f_i(\mu_i^1, \ldots, \mu_i^k)$, and $W_i = f_i(W_i^1, \ldots,
  W_i^k)$, for each $i=1,\ldots, n$.  Then, $(\bbb \varepsilon, \bbb \mu, \bbb
  W)$ is also semi-balanced, where $\bbb \mu = (\mu_1, \ldots, \mu_n)$
  and $\bbb W = (W_1, \ldots, W_n)$.
\end{proposition}

\begin{proof}
  First, we prove fairness for $\mu_i$: if $q_i = 0$, then
  $\mu_i^1(\bbb Q) = \ldots = \mu_i^k(\bbb Q)=0$ because, by
  assumption, each $\mu_i^j$ is fair.  Hence, $f_i(\mu_i^1(\bbb Q),
  \ldots, \mu_i^k(\bbb Q))=0$ because $f_i(\bbb 0)=0$.  Next, by
  \autoref{prop:composition}, each $\mu_i$ is arbitrage-free.
  Finally, each $\mu_i$ is compensating for $W_i$, because the
  functions $f_i$ are non-decreasing, and each $\mu_i^j$ is
  compensating for $W_i^j$, hence $f_i(\mu_i^1(\bbb Q), \ldots,
  \mu_i^k(\bbb Q)\!\geq\!f_i(W_i^1(\varepsilon_i(\mathcal{K}_{\bbb
    Q})), \ldots, W_i^k(\varepsilon_i(\mathcal{K}_{\bbb Q}))) =
  W_i(\varepsilon(\mathcal{K}_{\bbb Q}))$.
\end{proof}

We can use this proposition to design micro-payment functions that
allow the true private data of an individual to be disclosed, as in
\autoref{subsec:true:data}.  We illustrate this with an example.

\begin{example}
  Consider \autoref{ex:1}, where several voters give a rating in
  $\set{0,1,2,3,4,5}$ to each of two candidates $A$ and $B$.  Thus,
  $x_1, x_2$ represent the ratings of voter 1, $x_3, x_4$ of voter 2,
  etc.  Suppose voter 1 values her privacy highly, and would never
  accept a total disclosure: we choose linear contract functions
  $W_1(\varepsilon) = W_2(\varepsilon) = c \cdot \varepsilon$ for her
  two votes, and define the micro-payments as in \autoref{prop:basic},
  $\mu_i(\bbb Q) = \frac{6 \cdot c}{\sqrt{v/2}}|q_i|$ for $i=1,2$.  On
  the other hand, voter 2 is less concerned about her privacy, and is
  willing to sell the true values of her votes, at some high price $d
  > 0$: then we choose bounded contract functions $W_3(\varepsilon) =
  W_4(\varepsilon) = 2\cdot d/\Pi \cdot \texttt{atan}(\varepsilon)$
  (which is sub-additive, by \autoref{COR:FINITE}), and define the
  micro-payments accordingly, $\mu_i(\bbb Q) = 2\cdot d/\Pi \cdot
  \texttt{atan}(\frac{6}{\sqrt{v/2}}|q_i|)$, for $i=3,4$.  By
  \autoref{prop:f:mu} this function is also compensating and micro
  arbitrage-free, and, moreover, it is bounded by $\mu_i \leq d$,
  where the upper bound $d$ is reached by the total-disclosure query
  ($v=0$).
\end{example}

Finally, we choose a payment function such as to ensure that the
micro-payments are cost-recovering.

\begin{proposition} \label{prop:mu:sum} (1) Suppose that $(\bbb
  \varepsilon, \bbb \mu, \bbb W)$ is semi-\\balanced, and define
  $\pi(\bbb Q) = \sum_i \mu_i(\bbb Q)$.  Then, $(\pi, \bbb
  \varepsilon, \bbb \mu, \bbb W)$ is balanced.

  (2) Suppose that $(\pi, \bbb \varepsilon, \bbb \mu, \bbb W)$ is
  balanced and $\pi' \geq \pi$ is any arbitrage-free pricing function.
  Then $(\pi', \bbb \varepsilon, \bbb \mu, \bbb W)$ is also balanced.
\end{proposition}

\begin{proof}
  Claim (1) follows from \autoref{cor:compfunction} (the sum of
  arbitrage-free functions is also arbitrage-free), while claim (2) is
  straighforward.
\end{proof}

To summarize, the synthesis procedure for a pricing framework proceeds
as follows.  Start with the simple micro-payment functions given by
\autoref{prop:basic}, which ensure linear compensation for each
user.  Next, modify both the micro-payment and the contract functions
using \autoref{prop:f:mu}, as desired, in order to adjust to the
preferences of individual users, for example, in order to allow a user
to set a price for her true data.  Finally, define the query price to
be the sum of all micropayments (\autoref{prop:mu:sum}), then increase
this price freely, by using any method in \autoref{cor:compfunction}.

%
%



\section{Discussion} \label{sec:incentives}

In this section, we discuss two problems in pricing private data, and
show how they affect our pricing framework.  The first is how to
incentivize data owners to participate in the database and truthfully
report their privacy valuations, which is reflected in her contract
function $W_i$: this property is called {\em truthfulness} in
mechanism design.  The second concerns protection of the privacy
valuations itself, meaning that the contract $W_i$ may also leak
information to the buyer.


\subsection{Truthfulness}

How can we incentivize a user to participate, and to reveal her true
assessment for the privacy loss of a data item $x_i$?  All things being
equal, the data owner will quote an impossibly high price, for even a
tiny loss of her privacy.  In other words, she would choose a contract
function $W(\varepsilon)$ that is as close to $\infty$ as possible.

Incentivizing users to report their true valuation is a goal of
{\em mechanism design}.  This has been studied for private data only
in the restricted case of a single query, and has been shown to be a
difficult task.  Ghosh and Roth \cite{DBLP:conf/sigecom/GhoshR11}
 show that if the privacy valuations are sensitive, then it
is impossible to design truthful and individually rational direct
revelation mechanisms.
Fleischer et al circumvent this impossibility
result by assuming that the privacy valuation is drawn from known
probability distributions~\cite{DBLP:conf/sigecom/FleischerL12}.
Also, according to some experimental studies~\cite{Acquisti}, the
owner's valuation is often complicated and difficult for the owner to
articulate and different people may have quite different
valuations. Indeed, without a context or reference, it is hard for
people to understand the valuation of their private data. The design
of a truthful and private mechanism for private data, even for a
single query, remains an active research topic.


We propose a simpler approach, adopted directly from that introduced
by Aperjis and Huberman~\cite{DBLP:journals/firstmonday/AperjisH12}.
Instead of asking for their valuations, users are given a fixed number
of options.  For example, the users may be offered a choice between
two contract functions, shown in \autoref{fig:options}, which we call
Options A and B
(following~\cite{DBLP:journals/firstmonday/AperjisH12}):

\begin{description}
\item[Option A] For modest privacy losses,
  there is a small micro-payment, but for significant privacy losses there
  is a significant micro-payment.
\item[Option B] There is a non-zero micro-payment for even the smallest privacy
  losses, but even the maximal payment is much lower than that of
  Option A.
\end{description}

While these options were initially designed for a sampling-based query
answering mechanism~\cite{DBLP:journals/firstmonday/AperjisH12}, they
also work for our perturbation-based mechanism.  Risk-tolerant users
will typically choose Option A, while risk-averse users will choose
Option B.  Clearly, a good user interface will offer more than two
options; designing a set of options that users can easily understand
is a difficult task, which we leave to future work.


\begin{figure}[th]
  \centerline{\includegraphics[height=150pt]{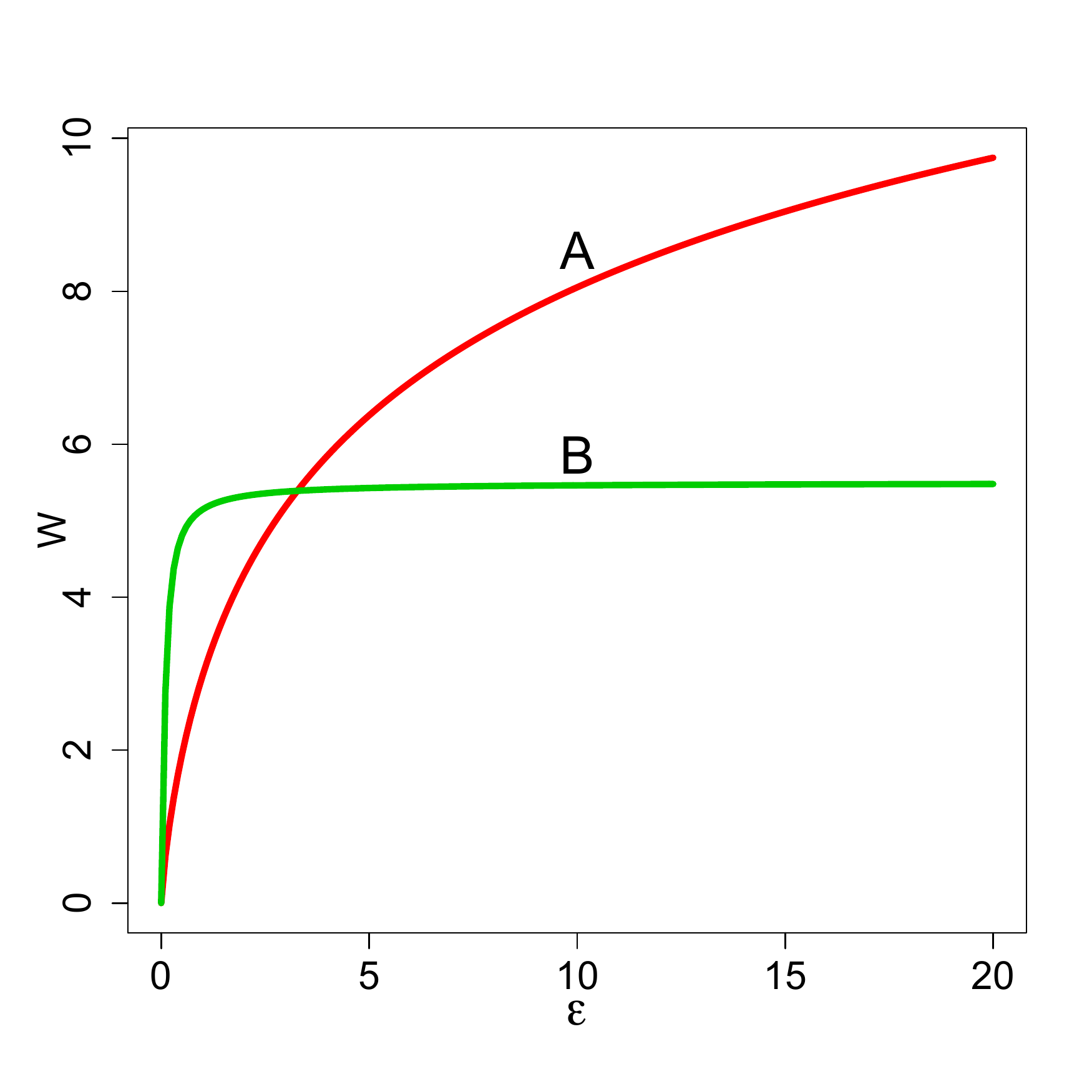}}
  \caption{\label{fig:options} Two options for the contract function
    $W$.  Option A makes a small micro-payment for small privacy losses
    and a large payment for large privacy losses.  Option B pays
    even for small privacy losses, but for large privacy
    losses pays less than A.  Risk-neutral users would typically choose
    Option A, while risk-averse users choose Option B.}
\end{figure}

\subsection{Private Valuations}

When users have sufficient freedom to choose their privacy valuation
(i.e.  their contract function $W_i$), then we may face another
difficult problem: the privacy valuation may be strongly correlated
with the data $x_i$ itself.  In that case, even releasing the price of a query may lead
to privacy loss, a factor not considered in our framework.  For example, consider a database of HIV status:
$x_i=1$ means that data owner $i$ has HIV, $x_i = 0$ means that she
does not.  Typically, users who have HIV will set a much higher value
on privacy of their $x_i$ than those who don't have HIV.  For example,
users without HIV may ask for \$1 for $x_i$, while users who do have
HIV may ask for \$1000.  Then, a savvy buyer may simply ask for the
price of a query, without actually purchasing the query, and determine
with some reasonable confidence whether a user has HIV.  Hiding the
valuation itself is a difficult problem, which is still being actively
researched in mechanism design~\cite{DBLP:conf/sigecom/FleischerL12}.

If the price itself is private, then inquires about prices need to be
perturbed in the same fashion as queries on the data.  Thus, the price
$\pi(\bbb Q)$ and the micro-payments $\mu_i(\bbb Q)$ need to be random
variables.  Queries are answered using a mechanism $\mathcal{K}$, while
prices are computed using a (possibly different) mechanism
$\mathcal{K}'$.  We show, briefly, that, if the contract functions are
linear $W_i = c_i \cdot \varepsilon_i$, then it is possible to extend
our pricing framework to ensure that data owners are compensated both
for the privacy loss from the query and the privacy loss from the
price function.  The properties of arbitrage-freeness, cost-recovery, and
compensation are now defined in terms of expected values. For example,
a randomized price function $\pi(\bbb Q)$ is arbitrage-free, if
$\set{\bbb Q_1,\ldots,\bbb Q_m}\rightarrow \bbb Q$ implies
$\EE{\pi(\bbb Q)} \leq \sum_{i=1}^m \EE{\pi(\bbb Q_i)}$.

Now the privacy loss for data item $x_i$ includes two parts. One part
is due to the release of the query answer, and the other part is due
to the release of the price.  Their values are
$\epsilon_i(\mathcal{K})$ and $\epsilon_i(\mathcal{K}')$ respectively.
A micropayment is {\em compensating} if $\EE{\mu_i(\bbb Q)} \geq c_i
\cdot (\varepsilon_i(\mathcal{K}) + \epsilon_i(\mathcal{K}'))$.

%

As for the data items, we assume that the constants $c_i$ used in the
contract function are drawn from a bounded domain $Y \subseteq \mathbb{R}$, and denote $\delta
= \text{sup}_{c \in Y} |c|$ (in analogy to $\gamma$ defined
in \autoref{sec:leaks}).  Assume that both $\mathcal{K}$ and
$\mathcal{K}'$ are Laplacian mechanisms. Given a query $\bbb Q = (\bbb
q, v)$ , set $b = \sqrt{v/2}$, choose some\footnote{When $b' \le
  \delta$, the expectation of the price $\pi$ is infinite.} $b' >
\delta$, tunable by the market maker. $\mathcal{K}$ is the mechanism
that, on an input $\bbb x$, returns $\bbb q(\bbb x) + \rho$, where
$\rho$ is a noise with distribution $Lap(b)$. $\mathcal{K}'$ is the
mechanism that, on an input $\bbb c$, returns a noisy price
$\frac{\gamma b'}{b\cdot (b' - \delta)} \sum_i c_i|q_i| + \rho'$,
where $\rho'$ is a noise with distribution $Lap(b')$. We denote the
exact price, $\frac{\gamma \cdot b'}{b\cdot (b' - \delta)} \sum_i
c_i\cdot |q_i|$, as $\EE{\mathcal{K}'(\bbb c)}$.  The sensitivity of
the mechanism $\mathcal{K}$ is $s_i(\mathcal{K}) = \gamma \cdot |q_i|$
(\autoref{def:person:sensitivity}).  If we define $s_i(\mathcal{K}') =
\frac{\gamma \cdot b'\cdot |q_i| \delta}{b \cdot (b' - \delta)}$, then
we prove (in the appendix):

\begin{align*}
  \epsilon_i(\mathcal{K}) \le \frac{s_i(\mathcal{K})}{b}, \quad  \epsilon_i(\mathcal{K}') \le \frac{s_i(\mathcal{K}')}{b'}.
\end{align*}

\begin{proposition} \label{PROP:GENERAL}
  Let $\mathcal{K}, \mathcal{K}'$ be Laplacian mechanisms (as
  described above) and $\mathbf{Q} = (\mathbf{q}, v)$ be a query. Set
  (as above), $b = \sqrt{v/2}$ and $b' > \delta$. Define:
\begin{align*}
\pi(\mathbf{Q}) = & \mathcal{K}'(\mathbf{c}) =  \EE{\mathcal{K}'(\bbb  c)} + \rho' \\
\mu_i(\mathbf{Q}) = & (\frac{s_i(\mathcal{K})}{b} + \frac{s_i(\mathcal{K}')}{b'}) \cdot c_i + \frac{\pi(\mathbf{Q})-  \EE{\mathcal{K}'(\bbb c)}}{n}, \\&\forall i=1,\ldots,n
\end{align*}
Then, $(\pi, \bbb \mu, \bbb \varepsilon, \bbb W)$ is a balanced mechanism.
\end{proposition}

\section{Related Work}  \label{sec:related}


Recent investigation of the tradeoff between privacy and utility in statistical databases was initiated by Dinur and Nissim \cite{DBLP:conf/pods/DinurN03}, and culminated in \cite{DBLP:conf/tcc/DworkMNS06}, where Dwork, McSherry, Nissim and Smith introduced \emph{differential privacy} and the \emph{Laplace mechanism}. The goal of this line of research is
to reveal accurate statistics while preserving the privacy of the individuals. There have been two (somewhat artificially divided) models involved: the non-interactive model, and the interactive model. In this paper, we use an interactive model, in which  queries arrive on-line (one at a time) and the market maker has to charge for them appropriately and answer them. There is a large and growing literature on differential privacy; we refer the readers to the recent survey by Dwork~\cite{DBLP:journals/cacm/Dwork11}. There is privacy loss in releasing statistics in a differentially private sense (quantified in terms of the privacy parameter/budget $\epsilon$). However, this line of research does not consider compensating the privacy loss.

Ghosh and Roth \cite{DBLP:conf/sigecom/GhoshR11} initiated a study of
how to incentivize individuals to contribute their private data and to
truthfully report their privacy valuation using tools of mechanism
design. They consider the same problem as we do, pricing private data,
but from a different perspective: there is only one query, and the
individuals' valuations of their data are private.  The goal is to
design a truthful mechanism for disclosing the valuation.  In
contrast, we assume that the individuals' valuations are public, and
focus instead on the issues arising from pricing multiple queries
consistently.  Another key difference is that we require not only accuracy
but also unbiasedness for the noisy answer to a certain query, while in
\cite{DBLP:conf/sigecom/GhoshR11} answers are not unbiased. There have
been some follow-ups to \cite{DBLP:conf/sigecom/GhoshR11},
e.g. \cite{DBLP:conf/sigecom/FleischerL12,
  DBLP:journals/corr/abs-1202-4741, DBLP:conf/sigecom/RothS12,
  DBLP:journals/corr/abs-1111-2885}; a good survey is
\cite{Roth:2012:BPD:2325713.2325714}. There are some other papers that
consider privacy and utility in the context of mechanism design,
e.g. \cite{DBLP:conf/sigecom/NissimOS12,
  DBLP:journals/corr/abs-1111-5472}.

Economic perspectives on the regulation and control of private information have a long history \cite{Stigler80Introduction,Posner81Economics}.  A national information market, where personal information could be bought and sold, was proposed by Laudon \cite{DBLP:journals/cacm/Laudon96}.  
Garfinkel et al.~\cite{DBLP:journals/tsmc/GarfinkelLNR06} proposed a methodology for releasing approximate answers to statistical queries and compensating contributing individuals as the basis for a market for private data.  That methodology does not use a rigorous measure of privacy loss or protection and 
does not address the problem of arbitrage.

Recently Balazinska, Howe and Suciu \cite{DBLP:journals/pvldb/BalazinskaHS11} initiated a study of data markets in the cloud (for general-purpose data, not specifically private data).  Subsequently, \cite{DBLP:conf/pods/KoutrisUBHS12} proposed a data pricing method which first sets explicit price points on a set of views and then computes the implied price for any query.  However, they did not consider the potential privacy risks of their method.  The query determinacy used in \cite{DBLP:conf/pods/KoutrisUBHS12} is instance-based, and as a result, the adversary could (in the worst case) learn the entire database solely by asking the prices of queries (for free).  Li and Miklau study data pricing for linear aggregation queries~\cite{lipricing} using a notion of instance-independent query determinacy.  This avoids some privacy risks, but it is still sometimes possible to infer query answers for which the buyer has not paid.  Both of the above works consider a model in which unperturbed query answers are exchanged for payment.  In this paper we consider noisy query answers and use an instance-independent notion of query determinacy, which allows us to formally model private disclosures and assign prices accordingly.

Aperjis and Huberman \cite{DBLP:journals/firstmonday/AperjisH12} describe a simple strategy to collect private data from individuals and compensate them, based on an assumption in sociology that some people are risk averse. By doing so, buyers could compensate individuals with relatively less money. More specifically, a buyer may access the private data of an individual with probability $0.2$, and offer her two choices: if the data is accessed, then she would be paid $\$10$, otherwise she would receive nothing; she would receive $\$1$ regardless whether her data would be used or not. Then a risk-averse person may choose the second choice, and consequently the buyer can save $\$1$ in expectation. In their paper, the private data of an individual is either entirely exposed, or completely unused. In our framework, there are different levels of privacy, the privacy loss is carefully quantified and compensated, and thus the data is better protected. Finally, Riederer et al. \cite{Riederer11sale:} propose auction methods to sell private data to aggregators, but an owner's data is either completely hidden or totally disclosed and the price of data is ultimately determined by buyers without consideration of owners' personal privacy valuations.


\section{Conclusions}

\label{sec:conclusions}

We have introduced a framework for selling private data.  Buyers can
purchase any linear query, with any amount of perturbation, and need
to pay accordingly.  Data owners, in turn, are compensated according to
the privacy loss  they incur for each query. In our framework
buyers are allowed to ask an arbitrary number of queries, and we have
designed techniques for ensuring that the prices are arbitrage-free,
meaning that buyers are guaranteed to pay for any information they
may further extract from the queries.  Our pricing framework is
balanced, in the sense that the buyer's price covers the
micro-payments to the data owner, and each micro-payment compensates
the users according to their privacy loss.

An interesting open question is whether we can achieve both truthfulness (as discussed in \cite{DBLP:conf/sigecom/GhoshR11}) and arbitrage-freeness (as discussed in the current paper) when pricing private data.

\paragraph*{Acknowledgements}
We appreciate the comments of each of the anonymous reviewers and, in particular, the suggestion of the example now presented in footnote 1 of \autoref{sec:basic}.  C. Li was supported by NSF CNS-1012748; Miklau was partially supported by NSF CNS-1012748, NSF CNS-0964094, and the European Research Council under the Webdam grant; D. Li and Suciu were supported by NSF IIS-0915054 and NSF CCF-1047815.

\newpage

\bibliographystyle{abbrv}
{ \bibliography{pricing-privacy-bib,relatedwork}}

\newpage

\appendix
\label{sec:appendix}

\section{Proof of \autoref{COR:FINITE}}
By \autoref{lemma:concave} it suffices to check that all first
derivatives are $\geq 0$ and all second derivatives are $\leq 0$, for
all $x \geq 0$:
\begin{align*}
&\frac{d}{dx}\texttt{atan}(x)=\frac{1}{1+x^2} > 0;\\
&\frac{d^2}{dx^2}\texttt{atan}(x)=-\frac{2x}{(1+x^2)^2} \leq 0;\\
&\frac{d}{dx}\texttt{tanh}(x)=\frac{1}{\texttt{cosh}^2(x)} > 0;\\
&\frac{d^2}{dx^2}\texttt{tanh}(x)=-\frac{2\texttt{tanh}(x)}{\texttt{cosh}^2(x)} \leq 0;\\
&\frac{d}{dx}\frac{x}{\sqrt{1+x^2}}=(1+x^2)^{-\frac{3}{2}} > 0;\\
&\frac{d^2}{dx^2}\frac{x}{\sqrt{1+x^2}}=-3x(1+x^2)^{-\frac{5}{2}} \leq 0.\\
&\end{align*}

\section{Proof of \autoref{PROP:GENERAL}}


We show that each $\mu_i$ is fair in expectation.  For individual
$i$, if $q_i = 0$, then by definition, $s_i(\mathcal{K}) = 0$ and
$s_i(\mathcal{K}') = 0$, and thus
$$\EE{\mu_i(\mathbf{q}, v)} = (\frac{s_i(\mathcal{K})}{b} + \frac{s_i(\mathcal{K}')}{b'}) \times c_i = 0$$.

We show that $\mu_i$ is micro arbitrage-free in expectation. For each
individual $i$, by definition,
\begin{align*}
\EE{\mu_i(\bbb Q)} & = \frac{\gamma b' \cdot c_i \cdot |q_i|}{b \cdot (b' - \delta)}\\
& = \frac{\sqrt{2}\gamma b' \cdot c_i}{b' - \delta}  \frac{|q_i|}{\sqrt{v}}.
\end{align*}

By the same argument as in \autoref{prop:basic}, $\EE{\mu_i(\bbb Q)}$
is arbitrage-free, and thus $\mu_i(\bbb Q)$ is arbitrage-free in
expectation.

We show that the micro-payments are cost recovering. By
definition,
\begin{align*}
\sum_i \mu_i(\bbb Q) &= \sum_i (\frac{s_i(\mathcal{K})}{b} + \frac{s_i(\mathcal{K}')}{b'}) \times c_i + \rho' \\
& = \sum_i (\frac{ \gamma |q_i|}{b} + \frac{\frac{\gamma b'|q_i|
    \delta}{b\cdot (b' - \delta)})}{b'}) \times c_i + \rho' \\
& = \frac{\gamma b'}{b\cdot (b' - \delta)} \sum_i c_i\cdot |q_i| + \rho' \\
& = \pi(\bbb Q),
\end{align*}
proving the claim.

Finally, we show that $\mu_i$ is compensating, in expectation: For
each individual $i$,
\begin{align*}
  \EE{\mu_i(\bbb Q)} & = (\frac{s_i(\mathcal{K})}{b} + \frac{s_i(\mathcal{K}')}{b'}) \times c_i\\
  & \ge (\varepsilon_i(\mathcal{K}) + \epsilon_i(\mathcal{K}')
  \times c_i,
\end{align*}
meaning that $\mu_i(\bbb Q)$ compensate user $i$ for her loss of privacy in expectation.

By a similar argument as in \autoref{prop:mu:sum}, $\pi(\bbb Q)$ is arbitrage-free in expectation.

\end{document}